\documentclass[12pt]{article}
\usepackage{mathrsfs}
\usepackage{dsfont}
\usepackage{amsthm}
\usepackage{mathrsfs}
\usepackage{amsmath}
\usepackage{amsfonts}
\usepackage[colorlinks,linkcolor=black,anchorcolor=blue,citecolor=blue]{hyperref}
\usepackage{amssymb, amsmath, cite}

\newtheorem{theorem}{Theorem}[section]
\newtheorem{lemma}{Lemma}[section]

\newcommand\be{\begin{equation}}
\newcommand\ee{\end{equation}}
\newcommand\ber{\begin{eqnarray}}
\newcommand\eer{\end{eqnarray}}
\newcommand\berr{\begin{eqnarray*}}
\newcommand\eerr{\end{eqnarray*}}
\newcommand\bea{\begin{eqnarray}}
\newcommand\eea{\end{eqnarray}}

\newcommand{\nn}{\nonumber}
\newcommand\vep{\varepsilon}\newcommand{\ii}{\mathrm{i}}
\newcommand{\dd}{\mathrm{d}}
\newcommand\e{\mathrm{e}}\newcommand\pa{\partial}

{\setlength\arraycolsep{2pt}

\begin{document}

\title{Vortices for the magnetic Ginzburg--Landau theory in curved space}
\author{Lei Cao, Yilu Xu, Shouxin Chen\\School of Mathematics and Statistics, Henan University,\\
Kaifeng, Henan 475004, PR China\\ }
\date{}
\maketitle

\begin{abstract}

Since the Ginzburg--Landau theory is concerned with macroscopic phenomena, and gravity affects how objects interact at the macroscopic level. It becomes relevant to study the Ginzburg--Landau theory in curved space, that is, in the presence of gravity. In this paper, some existence theorems are established for the vortex solutions of the magnetic Ginzburg--Landau theory coupled to the Einstein equations. First, when the coupling constant $\lambda=1$, we get a self--dual structure from the Ginzburg--Landau theory, then a partial differential equation with a gravitational term that has power--type singularities is deduced from the coupled system. To overcome the difficulty arising from the orders of singularities at the vortices, a constraint minimization method and a monotone iteration method are employed. We also show that the quantized flux and total curvature are determined by the number of vortices. Second, when the coupling constant $\lambda>0$, we use a suitable ansatz to get the radially symmetric case for the magnetic Ginzburg--Landau theory in curved space. The existence of the symmetric vortex solutions are obtained through combining a two--step iterative shooting argument and a fixed-point theorem approach. Some fundamental properties of the solutions are established via applying a series of analysis techniques.

\end{abstract}

\medskip
\begin{enumerate}

\item[]

\item[]
{MSC numbers(2020):} 35J57, 35Q75, 81T10.

\end{enumerate}

\section{Introduction}\label{s0}
\setcounter{equation}{0}

In 1950, Ginzburg and Landau established the basic dimensional theory of superconductivity, the Ginzburg--Landau theory. This is a "semi--quantum theory". In 1957, Abrikosov classified superconductors based on the Meissner effect within the framework of the Ginzburg--Landau theory. In the same year, Bardeen, Cooper and Schrieffer \cite{B,Ba} proposed a complete quantum theory of superconductivity, the BCS theory, based on Cooper's electron pairing theory. In 1959, Gorkov \cite{G} proved that the fundamental equation of Ginzburg--Landau theory, the Ginzburg--Landau equations, can be derived from BCS theory.
Since then, the literature on Ginzburg--Landau theory showcases its versatility in describing a wide array of physical phenomena, from phase transitions \cite{H,Sc} and superconductivity \cite{Ch,G,Pa,Tr} to domain walls \cite{Fa,Ma} and skyrmions \cite{Ag,Li}. The rigorous derivations and applications presented in these studies highlight the importance of Ginzburg--Landau theory in both theoretical and computational physics.

The Ginzburg--Landau energy has been the subject of numerous studies since the pioneering work of Bethuel, Brezis, and H\'{e}lein in \cite{Be}. A series of techniques and concepts were developed for Ginzburg--Landau theory, such as: renormalized energies, the Pohozaev identities, lower bounds for the energy in terms of the vortices, and stationarity conditions, etc. However, the model without magnetic field does not exhibit all the phenomena observed in superconductors. So it seems more interesting to study the full Ginzburg--Landau theory, which represents the magnetic field.
It is well known that vortices always repel each other. Like magnets with the same poles, there seems to be a natural repulsive force between them that keeps them away from each other. However, it's not that simple. In this physical system, the vortexes were originally separate, mutual exclusion and scattered around, but the magnetic field is like pulling up an invisible fence, they are confined to the centre of the domain. Besides, when an externally applied magnetic field intervenes in this physical system, it triggers a series of complex and wonderful changes. This applied magnetic field is like a special key that induces a phase transition in matter. In this physical context, the intervention of the magnetic field causes the internal structure and state of matter to change essentially. Moreover, the applied magnetic field has a unique selective power that determines the number of vortices. It is as if the magnetic field is selecting the number of vortices that should be present in the system according to a specific rule, and this precision and regularity is what is so fascinating in physics.

Gravity has no significant effect at the level of subatomic particles. However, at the macroscopic level, gravity is the most important interaction between objects, determining the motion of planets, stars, galaxies and even light.
Galaxy formation in early cosmology \cite{V,Vi} induces local concentrations of energy and gravitational curvature, a phenomenon that arises from the coupling of a suitable system of matter and gauge fields to Einstein equations.
As is well known, Ginzburg--Landau theory is used in physics to describe phase transitions and critical phenomena \cite{Mi,De} in macroscopic systems, and is closely related to macroscopic properties. The Ginzburg--Landau equations exhibit vortex solutions, which can be viewed as finite energy solutions to the equations describing the two--dimensional Abelian Higgs model.
Cosmic string \cite{Br,Kib,Vil,Wi,Ze} solutions of the Abelian Higgs model with conformal coupling to gravity \cite{Ve} have been explored, demonstrating the interplay between these two theories. Additionally, the coupling of the Abelian Higgs theory to New Massive Gravity \cite{F} has been investigated, revealing how superconductors relate to gravitation. Moreover, Abelian Higgs model has been coupled to gravity near a black hole horizon \cite{Gu}, and the implications of this coupling for gauge symmetry and cosmological constants have been highlighted. All these studies underline the significance of exploring the interplay between the Abelian Higgs theory and gravity in understanding the dynamics of quantum field theories.

The aim of this article is to study the full Ginzburg--Landau theory of vortices with gravity which will cause space to curve at infinity.
There is one coupling constant $\lambda$ in the Ginzburg--Landau theory, and $\lambda=1$ borders the two types of superconductors such that when $\lambda<1$, the superconductor is of type I, when $\lambda>1$, the superconductor is of type II.
In order to overcome the difficulty of studying the coupled system of the Einstein equations and the usual second-order Euler--Lagrange equations of the geometrized Ginzburg--Landau theory, we take the coupling constant $\lambda=1$, which allows us to obtain a self--dual system of the Ginzburg--Landau theory. Consequently, use the simplifies Einstein equations, we get the formulation of the gravitational metric factor $\e^\eta$.

The rest of this paper is organized as follows. In Section 2 and 3, we introduce the geometrization form of the magnetic Ginzburg--Landau theory and derive the corresponding self--dual system, which can actually be reduced to a nonlinear elliptic equation with a gravitational term of highly complex exponential form when the coupling constant $\lambda=1$. In Section 4, a monotone iteration method is employed to prove the existence of solutions to the nonlinear elliptic equation with topological boundary condition. In this process, we apply McOwen's method \cite{M} to study the existence of the lower solution. We then go on in Section 5 to study the quantized flux and total curvature of the magnetic Ginzburg--Landau theory with gravity. Finally, Section 6 gives the radially symmetric form of the magnetic Ginzburg--Landau theory in curved space when the coupling constant $\lambda>0$, and illustrates the existence of the symmetric vortex solutions. A series of its qualitative properties are also discussed.

\section{The magnetic Ginzburg--Landau theory with gravity}\label{s1}
\setcounter{equation}{0}

It is well known that in the static limit and the temporal gauge the Abelian Higgs theory reduces itself into the Ginzburg--Landau theory for superconductivity. Thus, with the gauge field $A_\mu=(0,A_1,A_2,0)$ and the Higgs field $\phi$, the Hamiltonian density assumes the form
\be\label{1.1}
\mathcal{H}=\frac{1}{2}F_{12}^2+\frac{1}{2}|D_1\phi|^2+\frac{1}{2}|D_2\phi|^2+\frac{\lambda}{8}(|\phi|^2-1)^2,
\ee
over $\mathbb{R}^2$ with the coordinates $x_1,x_2$ with $F_{12}=\pa_1 A_2-\pa_2 A_1$, $D_j\phi=\partial_j\phi-\ii A_j\phi$, $x=(x^j), j=1,2$, and a dimensionless coupling parameter $\lambda>0$.
Let the gravitational metric of the spacetime be
\be\label{1.5}
g_{\mu\nu}=\rm{diag}\{1,-\e^\eta,-\e^\eta,-1\}, \mu,\nu=0,1,2,3.
\ee
Then \eqref{1.1} is modified into a form with broken vacuum symmetry
\bea\label{1.4}
\mathcal{H}&=&\frac{1}{2}g^{\mu\mu'}g^{\nu\nu'}F_{\mu\nu}F_{\mu'\nu'}-\frac{1}{2}g^{\mu\nu}D_\mu \phi\overline{D_\nu \phi}+\frac{\lambda}{8}(|\phi|^2-1)^2,\nn\\
&=&\frac {1}{2}\e^{-2\eta}F_{12}^2+\frac {1}{2}\e^{-\eta}(|D_1 \phi|^2+|D_2 \phi|^2)+\frac{\lambda}{8} (|\phi|^2-1)^2.
\eea
Thus, the energy of the magnetic Ginzburg--Landau theory with gravity is
\be\label{1.4b}
E=\int_{\mathbb{R}^2}\mathcal{H}\e^\eta\dd x.
\ee
In the following, we aim to find a static energy--minimizing solution $(\phi, A)$ which is in fact a multiple vortex, so that the gravitational metric is of the form \eqref{1.5}.

Our main result of the existence theorem concerning multiple vortices is as follows.

\begin{theorem}\label{th1.1}
For given distinct prescribed points $p_1, p_2,\ldots,p_N\in \mathbb{R}^2$ satisfying
\be\label{1.4c}
4\pi GN\leq 1,
\ee
the energy \eqref{1.4b} has a finite--energy solution $(\phi, A_j)(A_0=A_3=0)$ with the associated two--dimensional gravitational metric $g_{jk}=\e^\eta\delta_{jk}$ which is determined explicitly by the formula
\be\nn
\e^\eta=g_0\left(\e^{u-\e^u}\prod_{s=1}^N|x-p_s|^{-2}\right)^{4\pi G},
\ee
where $g_0>0$ is a constant and $u=\ln|\phi|^2$. The zeros of $\phi$ are exactly $p_1, p_2,\ldots,p_N$ of unit charge and the quantized values of the magnetic flux, the energy of the matter--gauge sector and the total gravitational curvature are
\be\nn
\int_{\mathbb{R}^2}F_{12}\dd x=\pm\pi N,~~\int_{\mathbb{R}^2}\mathcal{H}\e^\eta\dd x=\pi N,~~\int_{\mathbb{R}^2}K_\eta\e^{\eta}\dd x=8\pi^2GN,
\ee
where $K_\eta$ is the Gauss curvature~\textup{(}we give its expression in \eqref{1.21}\textup{)}. Besides, as $|x|\to\infty$, the solution approaches the vacuum with broken symmetry at the rate
\bea
F_{12}&=&O(|x|^{-(b+8\pi GN)}),\nn\\
1-|\phi|^2&=&O(|x|^{-b}),\nn\\
|D_j\phi|&=&O(|x|^{-b_1}),\nn
\eea
where $b, b_1>0$ are arbitrary.
\end{theorem}

\section{Governing euqations}\label{s2}
\setcounter{equation}{0}

The energy--minimizing solution is hard to approach, while we can find a self--dual system to reduce the problem to a structure of nonlinear elliptic equation. To proceed, we introduce a current density
\be\label{1.6}
J_k=\frac{\ii}{2}(\phi \overline{D_k \phi}-\overline{\phi} D_k \phi),~~k=1,2.
\ee
Then the commutation relation for gauge--covariant derivatives
\be\label{1.7}
[D_j,D_k]\phi=(D_j D_k-D_k D_j)\phi=-\ii(\partial_j A_k-\partial_k A_j)\phi=-\ii F_{jk}\phi,~~j,k=1,2,
\ee
and the formula $\partial_j(u\overline{u})=\overline{u}D_j u+u \overline{D_j u}$ allow us to differentiate \eqref{1.6} as
\be\label{1.8}
J_{12}=\partial_1 J_2-\partial_2 J_1=\ii(D_1 \phi \overline{D_2 \phi}-\overline{D_1 \phi} D_2 \phi)-|\phi|^2 F_{12}.
\ee
Therefore, if we take $\lambda=1$, \eqref{1.4} can be rewritten as
\bea\label{1.9}
\mathcal{H}&=&\frac {1}{2}\e^{-2\eta}F_{12}^2+\frac {1}{2}\e^{-\eta}(|D_1 \phi|^2+|D_2 \phi|^2)+\frac{1}{8} (|\phi|^2-1)^2\nn\\
&=&\frac{1}{2}\left(\e^{-\eta}F_{12}\pm \frac {1}{2} (|\phi|^2-1)\right)^2+\frac {1}{2}\e^{-\eta}|D_1 \phi\pm \ii D_2 \phi|^2\nn\\
&&\pm\frac {1}{2}\e^{-\eta}\left(\ii(D_1 \phi \overline{D_2 \phi}-\overline{D_1 \phi} D_2 \phi)-F_{12}(|\phi|^2-1)\right)\nn\\
&=& \frac {1}{2}\left((\e^{-\eta}F_{12}\pm \frac {1}{2} (|\phi_1|^2-1))^2+\e^{-\eta}|D_1 \phi\pm \ii D_2 \phi|^2\right)\pm \frac 12\e^{-\eta}F_{12}\pm \frac 12\e^{-\eta}J_{12}.
\eea
From \eqref{1.9}, we are lead to the self--dual system
\bea
F_{12}&=&\mp\frac{1}{2}\e^\eta(|\phi|^2-1),\label{1.10}\\
D_1 \phi\pm \ii D_2 \phi&=&0.\label{1.11}
\eea
Assume that the set
\be\label{1.12}
P=\{p_1,p_2,\cdots,p_N\}
\ee
determines the locations of zeros (counting multiplicities) of $\phi$ and $N>1$. The equation \eqref{1.11} says in view of the $\overline{\pa}$--Poincar\'{e} lemma that, $\phi$ is holomorphic in a neighborhood around $p_s$, then locally
\be\label{1.13}
\phi(x)=h_{p_s}(x)(x-p_s), s=1,2...,N,
\ee
where $h_{p_s}$ is a locally well--defined nonvanishing function. The zeros of $\phi$ are obviously isolated and give the locations of the vortices of a solution. It is straightforward to examine that any solution of \eqref{1.10}--\eqref{1.11} necessarily leads to the lower energy bound $E=\int \mathcal{H}\e^\eta\dd x=\pi N$ where $N$ is positive and represents the total vortex number over $(\mathbb{R}^2,\e^\eta\delta_{ij})$.

Let $u=\ln|\phi|^2$, then the system \eqref{1.10}--\eqref{1.11} is reduced to
\be\label{1.14a}
\triangle_g u=(\e^u-1)+4\pi\sum_{s=1}^N \delta_{p_s}(x),~~x\in\mathbb{R}^2,
\ee
where $\triangle_g u$ is the Laplace--Beltrami operator induced from the metric $g=\{g_{ij}\}$ defined by
\be\label{1.14b}
\triangle_g u=\frac{1}{\sqrt{g}}\pa_j(g^{jk}\sqrt{g}\pa_k u).
\ee
Consequently, inserting \eqref{1.14b} into \eqref{1.14a}, we see that
\be\label{1.14}
\triangle u=\e^\eta(\e^u-1)+4\pi\sum_{s=1}^N \delta_{p_s}(x),~~x\in\mathbb{R}^2.
\ee
Here the unknown gravitational factor $\e^\eta$ is to be determined in the following. For this purpose, we consider the Einstein equations
\be\label{1.2}
G_{\mu\nu}=-8\pi G T_{\mu\nu},
\ee
where $G$ is Newton's gravitational constant, $T_{\mu\nu}$ represents energy momentum tensor, and $G_{\mu\nu}$ denotes the Einstein tensor with the form
\be\label{1.3}
G_{\mu\nu}=R_{\mu\nu}-\frac{1}{2}g_{\mu\nu}R,
\ee
where $g_{\mu\nu}$ is the gravitational metric of the spacetime with the Minkowski signature $(+---)$, $R$ the scalar curvature, and $R_{\mu\nu}$ the Ricci tensor.

The Gauss curvature $K_\eta$ can be calculated by the expression
\be\label{1.21}
K_\eta=-\frac{1}{2}\e^{-\eta}\triangle\eta
\ee
on the two--surface $(\mathbb{R}^2,\e^\eta\delta_{ij})$. We see that the metric \eqref{1.5} simplifies the Einstein equations to a singular scalar equation
\be\label{1.12}
K_\eta=8\pi G\mathcal{H}.
\ee
Using \eqref{1.10}--\eqref{1.11}, when away from zeros of $\phi$, the energy density $\mathcal{H}$ has the form
\bea\label{1.19}
\mathcal{H}&=&\frac {1}{2}\e^{-2\eta}F_{12}^2+\frac {1}{2}\e^{-\eta}(|D_1 \phi|^2+|D_2 \phi|^2)+\frac{\lambda}{8} (|\phi|^2-1)^2\nn\\
&=&\mp\frac{1}{2}\e^{-\eta}F_{12}(|\phi|^2-1)+\frac{1}{2}\e^{-\eta}(|D_1\phi|^2+|D_2\phi|^2)\nn\\
&=&\frac{1}{4}\e^{-\eta}(\e^u-1)\triangle u+\frac{1}{4}\e^{-\eta}\e^u|\nabla u|^2.
\eea
Taking the zeros $p_s$ into consideration, we deduce the relation at the full $\mathbb{R}^2$
\be\label{1.20}
4\e^\eta\mathcal{H}=\triangle(\e^u-u)+4\pi\sum_{s=1}^N \delta_{p_s}(x),~~x\in\mathbb{R}^2.
\ee
In view of \eqref{1.12}, \eqref{1.21} and \eqref{1.20}, we get a harmonic function
\be\label{1.22}
\frac{\eta}{4\pi G}+\e^u-u+\sum_{s=1}^N\ln|x-p_s|^2,
\ee
which we assume to be a constant. Thus the gravitational metric factor takes the form
\be\label{1.15}
\e^\eta=g_0\left(\e^{u-\e^u}\prod_{s=1}^N|x-p_s|^{-2}\right)^{4\pi G}.
\ee
Here $g_0>0$ is an arbitrary constant. Since we are interested in solutions in the broken symmetry category so that $|\phi|^2$=1 at infinity or
\be\label{1.16}
u(x)\to 0~~\text{as}~~|x|\to\infty.
\ee
Besides, the boundary condition \eqref{1.16} gives us that
\be\label{1.17}
\e^\eta=O(|x|^{-8\pi G N}),~~|x|\gg 1,
\ee
which leads to the deficit angle
\be\label{1.18}
\delta=8\pi^2 G N.
\ee
In fact, the deficit angle should be less than $2\pi$, so we deduce that $4\pi GN\leq 1$, which can also be verified by the geodesic completeness \cite{Cao} of the associated two-dimensional gravitational metric $g_{ij}=\e^\eta\delta_{ij}$.

\section{Construction of solutions}\label{s3}
\setcounter{equation}{0}

In this section we prove the existence of vortices of \eqref{1.14} with topological boundary condition. Recall the analysis of the previous chapter on the deficit angle and the geometric completeness, we see that the equation \eqref{1.14} only make sense when $4\pi GN\leq1$. To overcome the difficulty posed by singularities in the gravitational factor, we will study the equation \eqref{1.14} in two cases, depending on the value of $4\pi GN$. First we use a constrained minimization method to study the equation \eqref{1.14} for the case $2\pi GN>1$ and obtain a solution $u_1$. We next show that $u_1$ is a subsolution of \eqref{1.14} when $0<4\pi GN\leq1$. Then a monotone iteration method is employed to construct vortices. We finally establish asymptotic estimates for the solution obtained.

\subsection{Existence of vortices for $2\pi GN>1$}
We start from the case $2\pi GN>1$.
\begin{theorem}\label{th2.1}
For any given distinct prescribed vortex points $p_1, p_2,\ldots,p_N\in \mathbb{R}^2$, when $4\pi G<1$ and $2\pi GN>1$, we can find a constant $C_0$ satisfying
\be\nn
\int_{\mathbb{R}^2}\e^\eta\dd x=C_0+4\pi N~~\text{and}~~0<C_0<8\pi,
\ee
such that the equation \eqref{1.14} has a solution which vanishes at infinity.
\end{theorem}

It is clear that any solution $u$ of \eqref{1.14} that vanishes at infinity must be negative--valued, which can be examined from the maximum principle.

Define $u_0(x)$ to be a smooth function on $\mathbb{R}^2$ so that
\bea\label{2.3}
u_0(x)=
\begin{cases}
\ln|x-p_s|^2,&~~~~x~\text{near}~p_s, s=1,2,\cdots,N,\\
0,&~~~~x\gg 1.
\end{cases}
\eea
Then
\be\nn
\triangle u_0=4\pi\sum_{s=1}^N \delta_{p_s}(x)-g,~~g\in C_c^\infty(\mathbb{R}^2),~~\int_{\mathbb{R}^2}g\dd x=4\pi N.
\ee
We now introduce the substitution $u=u_0+v$ in \eqref{1.14}, then $v$ satisfies
\be\label{2.4}
\triangle v=\e^\eta(\e^{u_0+v}-1)+g,~~v(x)\to 0~\text{as}~|x|\to\infty,
\ee
where
\be\label{2.5}
\e^\eta=g_0\exp(4\pi G(u-u_0-\e^{u-u_0}))\prod_{s=1}^N|x-p_s|^{-8\pi G}.
\ee
Choose $\delta>0$ so that
\be\nn
B_\delta(p_i)=\{x||x-p_i|<\delta\},~~i=1,2,\cdots,N
\ee
are disjoint, that is, $B_\delta(p_i)\cap B_\delta(p_i)=\emptyset, i\neq j$.
Let $B_R$ be a sufficiently large domain centred at $0$ with radius $R$ and satisfies $B_R\supset\bigcup_{i=1}^{N}B_\delta(p_i)$. Note that the decay estimate \eqref{1.17}, when $4\pi G<1$ and $2\pi GN>1$, there holds
\bea
\int_{\mathbb{R}^2}\e^\eta \dd x&=&(\int_{B_R}+\int_{\mathbb{R}^2-B_R})\e^\eta \dd x\nn\\
&=&C_1+\int_{\mathbb{R}^2-B_R}O(|x|^{-8\pi GN}) \dd x\nn\\
&=&C_1+2\pi\int_R^\infty O(r^{-8\pi GN})r\dd r\nn\\
&=&C_1+2\pi C_2\frac{R^{2-8\pi GN}}{8\pi GN-2}\equiv C.\nn
\eea
Clearly, $\e^\eta-g$ is of compact support and
\be\nn
\int_{\mathbb{R}^2}(\e^\eta-g)\dd x=C-4\pi N \equiv C_0.
\ee
Note the structure of $\e^\eta$, we can choose an appropriate small $g_0$ such that $C_0<8\pi$, which is required in Lemma \ref{le.4}.

As in McOwen \cite{M} we define the functionals
\bea
I(v)&=&\int_{\mathbb{R}^2}\left(\frac{1}{2}|\nabla v|^2+(g-\e^\eta)v\right)\dd x,\label{2.6}\\
J(v)&=&\int_{\mathbb{R}^2}\e^{\eta+u_0+v}\dd x.\label{2.7}
\eea
To proceed, we need to consider a suitable weighted Sobolev space formalism. Define the weighted measure $\dd \mu=\e^\eta\dd x$, where $\e^\eta$ is as in \eqref{2.5} and
\be\label{2.8}
\e^\eta=O(r^{-8\pi GN}),~~2\pi GN>1,~~r=|x|\gg 1.
\ee
Use the notation $L^p(\dd\mu)=L^p(\mathbb{R}^2,\dd\mu)$ and let $\mathscr{H}$ be the Hilbert space of $L_{loc}^2$ functions under the norm given by
\be\label{2.9}
||v||_{\mathscr{H}}^2=||\nabla v||_{L^2(\dd x)}^2+||v||_{L^2(\dd\mu)}^2<\infty.
\ee
Then $\mathscr{H}$ contains all constant functions and thus $v\mapsto \int_{\mathbb{R}^2} v\dd\mu$ is a continuous linear functional on $\mathscr{H}$ so that
\be\label{2.10}
\dot{\mathscr{H}}=\{v\in \mathscr{H}|\int_{\mathbb{R}^2} v\dd\mu=0\}
\ee
is a closed subspace of $\mathscr{H}$. Therefore we have for each $v\in\mathscr{H}$ the decomposition
\be\label{2.11}
v=\bar{v}+\dot{v},~~\bar{v}\in\mathbb{R},~~\dot{v}\in\dot{\mathscr{H}}.
\ee

First, we give some results established in McOwen \cite{M}.
\begin{lemma}\label{le.1}
For any $0<\vep<4\pi$, there is some constant $C(\vep)>0$ so that
\be\label{2.12}
\int_{\mathbb{R}^2}\exp(a|v|)\dd\mu\leq C(\vep)\exp\left(\frac{a^2}{4(4\pi-\vep)}||\nabla v||_{L^2(\dd x)}^2\right),~~v\in\dot{\mathscr{H}}, a\in\mathbb{R}.
\ee
\end{lemma}

\begin{lemma}\label{le.2}
There is a constant $C>0$ so that
\be\label{2.13}
||v||_{L^2(\dd\mu)}^2\leq C||\nabla v||_{L^2(\dd x)}^2,~~v\in\dot{\mathscr{H}}.
\ee
\end{lemma}

\begin{lemma}\label{le.3}
The injection $\mathscr{H}\mapsto L^2(\dd\mu)$ is a compact embedding.
\end{lemma}

From the above statement we see that both $I(v)$ and $J(v)$ are well defined on $\mathscr{H}$. Thus we are prepared to prove the existence of a solution to the vortex equation \eqref{2.4} when $2\pi GN>1$. Consider the optimization problem
\be\label{2.14}
\min\{I(v)|J(v)=C_0, v\in\mathscr{H}\}.
\ee

\begin{lemma}\label{le.4}
When $2\pi GN>1$, the problem \eqref{2.14} has a solution provided $0<C_0<8\pi$.
\end{lemma}
\begin{proof}
If $J(v)=C_0$, with the decomposition given in \eqref{2.11}, we have
\be\nn
\e^{\bar{v}}\int_{\mathbb{R}^2}\e^{u_0+\dot{v}}\dd\mu=C_0,
\ee
or
\be\label{2.15}
\bar{v}=\ln C_0-\ln\int_{\mathbb{R}^2}\e^{u_0+\dot{v}}\dd\mu.
\ee
Inserting \eqref{2.15} into \eqref{2.6}, we get
\bea
I(v)&=&\int_{\mathbb{R}^2}\left(\frac{1}{2}|\nabla \dot{v}|^2+g\bar{v}+g\dot{v}\right)\dd x-\int_{\mathbb{R}^2}v\dd\mu\nn\\
&=&\frac{1}{2}||\nabla\dot{v}||_{L^2(\dd x)}+\int_{\mathbb{R}^2}g\dot{v}\dd x-C_0\left(\ln C_0-\ln\int_{\mathbb{R}^2}\e^{u_0+\dot{v}}\dd\mu\right).\label{2.16}
\eea
Furthermore, by Lemma \ref{le.1}, we find
\be\label{2.17}
\int_{\mathbb{R}^2}\e^{u_0+\dot{v}}\dd\mu\leq C_3\int_{\mathbb{R}^2}\e^{\dot{v}}\dd\mu\leq C_3'\exp\left({\frac{1}{4\pi(4\pi-\vep)}||\nabla\dot{v}||_{L^2(\dd x)}^2}\right),
\ee
and
\be\label{2.18}
\left|\int_{\mathbb{R}^2}g\dot{v}\dd x\right|=\left|\int_{\mathbb{R}^2}g\e^{-\frac{\eta}{2}}e^{\frac{\eta}{2}}\dot{v}\dd x\right|\leq \vep^{-1}C_4+\vep C_4'||\dot{v}||_{L^2(\dd x)}^2.
\ee
Consequently, in view of \eqref{2.16}--\eqref{2.18} we yield the lower bound
\be\label{2.19}
I(v)\geq\frac{1}{2}\left(1-\frac{C_0}{2(4\pi-\vep)}-\vep C'\right)||\nabla\dot{v}||_{L^2(\dd x)}^2-C''(\vep),
\ee
where $C'$ is a constant independent of $\vep$.

Since $8\pi>C_0>0$, we can fix $\vep>0$ sufficiently small to make $1-\frac{C_0}{2(4\pi-\vep)}-\vep C'>0$. Then the right--hand side of \eqref{2.19} is bounded from below, that is, $I_0\equiv\min\{I(v)|J(v)=C_0, v\in\mathscr{H}\}$ is well defined. Let $\{v_j\}\in \mathscr{H}$ such that $I(v_j)\to I_0$. Then \eqref{2.19} indicates that $\{||\nabla\dot{v_j}||_{L^2(\dd x)}^2\}$ is a bounded sequence. Besides, \eqref{2.15} and \eqref{2.17} say that $\{\bar{v_j}\}$ is bounded as well. Thus, we may assume $\dot{v_j}\to\dot{v}$ weakly in $\dot{\mathscr{H}}$ and $\bar{v_j}\to\bar{v}$ in $\mathbb{R}$ as $j\to\infty$. Hence, by Lemma \ref{le.3}, without loss of generality, we may assume that $v_j\to v=\bar{v}+\dot{v}\in\mathscr{H}$ strongly in $L^2(\dd\mu)$. Therefore,
\be\nn
\left|\int_{\mathbb{R}^2}g v_j\dd x-\int_{\mathbb{R}^2}g v\dd x\right|\leq\int_{\mathbb{R}^2}|g|\e^{-\frac{\eta}{2}}|v_j-v|\e^{\frac{\eta}{2}}\dd x\leq C||v_j-v||_{L^2(\dd\mu)}\to 0,~~j\to\infty,
\ee
\be\nn
\left|\int_{\mathbb{R}^2}\e^\eta v_j\dd x-\int_{\mathbb{R}^2}\e^\eta v\dd x\right|\leq\int_{\mathbb{R}^2}|v_j-v|\dd\mu\leq C||v_j-v||_{L^2(\dd\mu)}\to 0,~~j\to\infty,
\ee
and
\bea
&&\left|\int_{\mathbb{R}^2}\e^{\eta+u_0+v_j}\dd x-\int_{\mathbb{R}^2}\e^{\eta+u_0+v}\dd x\right|\nn\\
&\leq&\int_{\mathbb{R}^2}\e^{\eta+u_0}|\e^{v_j}-\e^v|\dd x\nn\\
&\leq&\int_{\mathbb{R}^2}\e^{\eta+u_0+|v_j|+|v|}|v_j-v|\dd x\nn\\
&\leq&C\int_{\mathbb{R}^2}\e^{\eta+u_0+|\dot{v_j}|+|\dot{v}|}\e^{-\frac{3}{4}\eta}\e^{\frac{1}{4}\eta}|v_j-v|\e^{\frac{1}{2}\eta}\dd x\nn\\
&\leq&C\left(\int_{\mathbb{R}^2}\e^{4(\eta+u_0+|\dot{v_j}|)}\e^{-3\eta}\dd x\right)^{\frac{1}{4}}\left(\int_{\mathbb{R}^2}\e^{4|\dot{v}|}\dd\mu\right)^{\frac{1}{4}}\left(\int_{\mathbb{R}^2}|v_j-v|^2\dd\mu\right)^{\frac{1}{2}}\nn\\
&\leq&C'\exp\left(\frac{1}{4\pi-1}||\nabla\dot{v}||_{L^2(\dd x)}^2\right)||v_j-v||_{L^2(\dd\mu)}\to 0,~~j\to\infty.\nn
\eea
Thus we see that the functional $I(\cdot)$ is weakly lower semicontinuous on $\mathscr{H}$ and
\be\nn
I(v)\leq\liminf_{j\to\infty}I(v_j),~~J(v)=\lim_{j\to\infty}J(v_j)=C_0. \ee
In other words $v$ is a solution of \eqref{2.14}.
\end{proof}

\begin{lemma}\label{le.5}
The minimizer $v$ of \eqref{2.14} obtained in Lemma \ref{le.4} is a solution of \eqref{2.4} when $2\pi GN>1$.
\end{lemma}
\begin{proof}
Let $v$ be the solution of \eqref{2.14} obtained in Lemma \ref{le.4}, then $\exists \lambda\in\mathbb{R}$ so that
\be\label{2.20a}
\int_{\mathbb{R}^2}(\nabla v\cdot\nabla\chi+(g-\e^\eta)\chi)\dd x=\lambda\int_{\mathbb{R}^2}\e^{\eta+u_0+v}\chi\dd x,~~\forall\chi\in\mathscr{H}.
\ee
Inserting the test function $\chi\equiv1$ in \eqref{2.20a}, we get
\be\nn
-C_0=\lambda C_0.
\ee
Hence $\lambda=-1$, and $v$ is a weak solution of \eqref{2.4} when $2\pi GN>1$. The elliptic regularity theorem then implies that $v$ is a $C^\infty$ solution of \eqref{2.4} on $\mathbb{R}^2\setminus P$ when $2\pi GN>1$.
\end{proof}

Returning to the original variable $u=u_0+v$, we get a solution of \eqref{1.14} when $2\pi GN>1$.

\subsection{Existence of vortices}
In this subsection, we will use the result established in Theorem \ref{th2.1} to find a physically meaningful solution to \eqref{1.14} subject to the boundary condition $u(x)\to 0$ as $|x|\to\infty$, namely the solution when $4\pi GN\leq1$.

Here is our existence result for vortices.
\begin{theorem}\label{th2.2}
For any given distinct prescribed vortex points $p_1, p_2,\ldots,p_N\in \mathbb{R}^2$ and $0<4\pi GN\leq 1$, the equation \eqref{1.14} has a negative solution which vanishes at infinity with the sharp decay estimates
\bea
-C_b|x|^{-b}< u(x)< 0,~~~|x|>r_0,\label{2.1}\\
|\nabla u|< C_{b_1}|x|^{-b_1},~~~|x|>r_0,\label{2.2}
\eea
where $r_0>\max\{|p_s|~|~s=1,2,\cdots,N\}$ and $b,b_1>0$ are arbitrary.
\end{theorem}

Suppose $u_1$ is a solution obtained in Theorem \ref{th2.1}, then
\be\label{2.19a}
\triangle u_1=\e^{\eta_1}(\e^{u_1}-1)+4\pi\sum_{s=1}^N \delta_{p_s}(x),
\ee
where
\be\nn
\e^{\eta_1}=g_0\left(\e^{u_1-\e^{u_1}}\prod_{s=1}^N|x-p_s|^{-2}\right)^{4\pi G},~~2\pi GN>1.
\ee

\begin{lemma}\label{le.6}
$u_+=0$ is a supersolution of \eqref{1.14}.
\end{lemma}
\begin{proof}
In fact
\be\nn
0=\triangle u_+<\e^{\eta_+}(\e^{u_+}-1)+4\pi\sum_{s=1}^N \delta_{p_s}(x),
\ee
where
\be\nn
\e^{\eta_+}=g_0\left(\e^{u_+-\e^{u_+}}\prod_{s=1}^N|x-p_s|^{-2}\right)^{4\pi G}.
\ee
Hence, the lemma follows immediately.
\end{proof}

\begin{lemma}\label{le.7}
$u_-=u_1$ is a subsolution of \eqref{1.14} when $0<2\pi GN\leq1$.
\end{lemma}
\begin{proof}
Since $u_1<0$, we have
\be\nn
\triangle u_1=\e^{\eta_1}(\e^{u_1}-1)+4\pi\sum_{s=1}^N \delta_{p_s}(x)\geq\e^{\eta}(\e^{u_1}-1)+4\pi\sum_{s=1}^N \delta_{p_s}(x),
\ee
where
\be\nn
\e^{\eta}=g_0\left(\e^{u_1-\e^{u_1}}\prod_{s=1}^N|x-p_s|^{-2}\right)^{4\pi G},~~0<2\pi GN\leq1.
\ee
Thus, $u_1$ is a subsolution of \eqref{1.14} when $0<2\pi GN\leq1$.
\end{proof}

In fact, the physical background and geometric completeness of the gravitational factor $\e^\eta$ tell us that $0<4\pi GN\leq1$ in \eqref{1.14}. Hence, $u_1$ is a subsolution of \eqref{1.14}. Using the theorem established by Ni \cite{Ni} (Theorem 2.10), we can derive a solution $u$ of \eqref{1.14} in $\mathbb{R}^2$ satisfying
\be\nn
u_1(x)\leq u(x)\leq0.
\ee
Therefore, we get a solution $u(x)$ of \eqref{1.14}.

We now present the decay estimates for the solution of \eqref{1.14}. Choose $r_0>0$ sufficiently large, so that
\be\nn
P=\{p_1,p_2,\cdots,p_N\}\subset B(r_0)=\{x\in\mathbb{R}^2~|~|x|<r_0\}.
\ee
Then \eqref{1.14} becomes
\be\label{2.20}
\triangle u=\e^\eta(\e^u-1)~~\text{in}~~\mathbb{R}^2\setminus\overline{B(r_0)},
\ee
where
\be\label{2.21}
\e^\eta=O(|x|^{-8\pi GN})~~\text{and}~~0<4\pi GN\leq1.
\ee

\begin{lemma}\label{le.8}
Suppose $0<4\pi GN\leq 1$, then for any $b>0$, there is a suitable constant $C_b>0$ so that the solution of \eqref{1.14} has the bounds
\be\nn
-C_b|x|^{-b}< u(x)< 0,~~|x|>r_0,
\ee
where $b>0$ is arbitrary.
\end{lemma}
\begin{proof}
Taking the comparison function
\be\label{2.22}
w(x)=C|x|^{-b}.
\ee
Then
\be\label{2.23}
\triangle w=b^2r^{-2}w,~~|x|=r.
\ee
Choose $\xi\in[0,1]$ so that $\e^u-1=\e^{\xi u}u$. Then, combining this relation with \eqref{2.20}--\eqref{2.21} and \eqref{2.23}, we get
\be\label{2.24}
\triangle(u+w)=\e^{\eta+\xi u}u+b^2r^{-2}w<C_br^{-2}(u+w),~~|x|=r>r_0,
\ee
where $C_b$ is a constant depending on $b$. For such fixed $r_0$, we can take the constant in \eqref{2.22} large to make
\be\nn
\left(u(x)+w(x)\right)\mid_{|x|=r_0}>0.
\ee
Applying the maximum principle and the boundary condition $u=w=0$ at infinity to the differential inequality \eqref{2.24}, we have $u+w>0$ in $\mathbb{R}^2\setminus\overline{B(r_0)}$, which establishes the bound as expected.
\end{proof}

\begin{lemma}\label{le.9}
Suppose $0<4\pi GN\leq1$. Then for any $b>0$, there is a suitable constant $C_b>0$ so that the solution of \eqref{1.14} has the bounds
\be\nn
|\nabla u(x)|\leq C_b|x|^{-b},~~|x|>r_0,
\ee
where $b>0$ is arbitrary.
\end{lemma}
\begin{proof}
Let $U=\pa_j u$ and differentiating \eqref{2.20}, we obtain
\be\label{2.25}
\triangle U=(\pa_j\eta)\e^\eta(\e^u-1)+\e^{\eta+u}U.
\ee
Clearly, the right--hand side of \eqref{2.25} lies in $L^2(\mathbb{R}^2\setminus\overline{B(r_0)})$, use the elliptic $L^2$--estimates, we have $U\in W^{2,2}(\mathbb{R}^2\setminus\overline{B(r_0)})$, which means $|\nabla u|\to 0$ as $r=|x|\to\infty$. Thus, there holds
\be\label{2.26}
\triangle U=H(u)r^{-8\pi GN}U+O(r^{-8\pi GN-b-1}),~~|x|=r\to\infty,
\ee
where $H(u)\to$ some positive constant as $u\to 0$ and the value of $b$ is as given in Lemma \ref{le.8}.

For the function defined by
\be\label{2.27}
w(x)=C|x|^{-b_1},
\ee
we have
\bea\label{2.28}
\triangle w&=&b_1^2r^{-2}w=(b_1^2+1)r^{-2}w-Cr^{-2-b_1}\nn\\
&\leq& H(u)r^{-8\pi GN}w-Cr^{-2-b_1},~~|x|=r>r_0,
\eea
where $0<4\pi GN<1$ and $r_0$ is sufficiently large. If $4\pi GN=1$, we can choose a suitable $g_0$ in \eqref{1.15} to make \eqref{2.28} holds.

Thus, using \eqref{2.26} and \eqref{2.28}, we obtain
\be\label{2.29}
\triangle(U-w)\geq H(u)r^{-8\pi GN}(U-w)+(Cr^{-2-b_1}+O\left(r^{-8\pi GN-b-1})\right),
\ee
where $b$ can be made as large as we please due to Lemma \ref{le.8}. Then, for $2+b_1<8\pi GN+b+1$, the inequality \eqref{2.29} gives us
\be\label{2.30}
\triangle(U-w)\geq H(u)r^{-8\pi GN}(U-w),~~|x|=r>r_0.
\ee
Choose $C$ in \eqref{2.27} large enough so that $U(r_0)-w(r_0)\leq 0$. Using this and that $U-w$ vanishes at infinity in \eqref{2.30}, we obtain $U\leq w$ for $|x|>r_0$ by the maximum principle.

On the other hand, for $|x|$ sufficiently large and $0<4\pi GN\leq1$, we have
\be\label{2.31}
\triangle(U+w)\leq H(u)r^{-8\pi GN}(U+w)-(Cr^{-2-b_1}+O\left(r^{-8\pi GN-b-1})\right),
\ee
which means
\be\label{2.32}
\triangle(U+w)\leq H(u)r^{-8\pi GN}(U+w)
\ee
for $2+b_1<8\pi GN+b+1$. Then we can taking $C$ in \eqref{2.27} large enough to make $U+w\geq 0$ for $|x|=r_0$. Consequently, the maximum principle gives us $U>-w$ for $|x|>r_0$.
\end{proof}
\medskip

\section{Quantized flux and total curvature}\label{s4}
\setcounter{equation}{0}

Using the solution $u$ constructed in Theorem \ref{th2.2}, we can get a solution pair $(\phi, A)$ of the coupled Einstein and Ginzburg--Landau theory \eqref{1.10}--\eqref{1.11} by the following standard prescription
\bea
\phi(z)&=&\exp\left(\frac{1}{2}u(z)+\ii \sum_{s=1}^N \arg(z-p_s)\right),~~z=x^1+\ii x^2,\label{3.1}\\
A_1(z)&=&-\text{Re}\{2\ii \bar{\pa}\ln \phi(z)\},~~A_2(z)=-\text{Im}\{2\ii \bar{\pa}\ln \phi(z)\}.\label{3.2}
\eea
Hence, if the local representation of $\phi$ is $\phi=\e^{\sigma+\ii w}$ where $\sigma$ and $w$ are real functions, by $\pa_1=\pa+\bar{\pa}, \pa_2=\ii(\pa-\bar{\pa})$ and \eqref{3.2}, we get
\bea
D_1 \phi&=&(\pa+\bar{\pa})\phi-(\frac{\bar{\pa}\phi}{\phi}-\frac{\pa\bar{\phi}}{\bar{\phi}})\phi=2\phi\pa\sigma,\label{3.3}\\
D_2 \phi&=&\ii(\pa-\bar{\pa})\phi+\ii(\frac{\bar{\pa}\phi}{\phi}+\frac{\pa\bar{\phi}}{\bar{\phi}})\phi=2\ii \phi\pa\sigma,\label{3.4}
\eea
which enable us to obtain the relation
\be\label{3.5}
|D_1 \phi|^2+|D_2 \phi|^2=\frac 12\e^u|\nabla u|^2.
\ee
Let $u=u_0+v$, then the equation \eqref{1.14} is transformed into \eqref{2.4}. Recall the decay estimate of $u$ stated in Theorem \ref{th2.2}, we have
\be\label{3.6}
\int_{\mathbb{R}^2}\triangle v=\lim_{r\to\infty}\oint_{|x|=r}\frac{\pa v}{\pa n}\dd s=0.
\ee
Thus, integrating \eqref{2.4} and use the equation \eqref{1.10}, we see that the magnetic flux is
\be\label{3.7}
\Phi=\int_{\mathbb{R}^2}F_{12}=\mp\frac{1}{2}\int_{\mathbb{R}^2}\e^\eta(|\phi|^2-1)=\mp\frac{1}{2}\int_{\mathbb{R}^2}\e^\eta(\e^u-1)=\pm\frac{1}{2}\int_{\mathbb{R}^2}g=\pm2\pi N.
\ee
Besides, as $|x|\to\infty$, we have the following estimates
\bea
F_{12}&=&O(|x|^{-(b+8\pi GN)}),\label{3.8}\\
1-|\phi|^2&=&O(|x|^{-b}),\label{3.9}\\
|D_j\phi|&=&O(|x|^{-b_1}),\label{3.10}
\eea
where $b$ and $b_1$ are the exponents described in Theorem \ref{th2.2}.

Furthermore, the energy density of the matter--gauge sector can be rewritten as follows
\bea
\mathcal{H}&=&\frac {1}{2}\e^{-2\eta}F_{12}^2+\frac {1}{2}\e^{-\eta}(|D_1 \phi|^2+|D_2 \phi|^2)+\frac {1}{8} (|\phi|^2-1)^2\nn\\
&=&\frac{1}{2}\left(\e^{-\eta}F_{12}\pm \frac {1}{2} (|\phi|^2-1)\right)^2+\frac {1}{2}\e^{-\eta}|D_1 \phi\pm \ii D_2 \phi|^2\nn\\
&&\pm\frac {1}{2}\e^{-\eta}\left(\ii(D_1 \phi \overline{D_2 \phi}-\overline{D_1 \phi} D_2 \phi)-F_{12}(|\phi|^2-1)\right)\nn\\
&=&\pm \frac 12\e^{-\eta}F_{12}\pm\frac 12\e^{-\eta}\text{Im}\{\pa_j\epsilon^{jk}\overline{\phi}(D_k \phi)\},\label{3.11}
\eea
thus, using \eqref{3.10} and \eqref{3.11}, we see the total energy of the matter--gauge coupling is
\be\label{3.12}
\int_{\mathbb{R}^2}\mathcal{H}\e^{\eta}\dd x=\pi N.
\ee
Recalling the Einstein equations \eqref{1.12}, we are led to the energy of the gravitational sector, which is realised by the total Gauss curvature
\be\label{3.13}
\int_{\mathbb{R}^2}K_\eta\e^{\eta}\dd x=8\pi^2GN.
\ee
Since the energy $E$ is written as
\be\label{3.14}
E=\int_{\mathbb{R}^2}\mathcal{H}\e^{\eta}\dd x=\int_{\mathbb{R}^2}F_{12}+J_{12},
\ee
inserting \eqref{3.7} and \eqref{3.12} into \eqref{3.14}, we obtain the current flux generated from $J_{12}$
\be\label{3.15}
\int_{\mathbb{R}^2}J_{12}=\pi N\pm2\pi N.
\ee

\section{Symmetric solutions}\label{s5}
\setcounter{equation}{0}

The goal of this section is to study the existence of the vortex solutions for the energy functional \eqref{1.4b} of the magnetic Ginzburg--Landau theory with gravity in radially symmetric case when coupling constant $\lambda>0$. A two--step iterative shooting argument and a fixed-point theorem approach are employed to obtain the radially symmetric vortex solutions. Besides, some qualitative properties for the solutions are established.

\subsection{Equations and functional in radial case}

If we set $r=|x|$ and suppose that the radially symmetric $N$--vortex solutions $(\phi, A)$ centered at the origin of the energy \eqref{1.4b} in terms of polar coordinates $(r,\theta)$ on $\mathbb{R}^{2}\backslash\{0\}$ are of the form
\ber\label{4.1}
\phi(r,\theta)=u(r)\mathrm{e}^{\mathrm{i}N\theta},~A_{i}=Nv(r)\varepsilon_{ij}\frac{x^{j}}{r^{2}},~i,j=1,2~\text{and}~\eta=\eta(r),
\eer
then \eqref{1.4b} can be reduced to
\ber\label{4.4}
&&E(u,v)=\frac{1}{2}\int_{\mathbb{R}^{2}}\bigg\{(u'(r))^{2}+\left(\frac{v'(r)}{r}\right)^{2}N^{2}\mathrm{e}^{-\eta(r)}+\frac{N^{2}}{r^{2}}u^{2}(r)\left(v(r)-1\right)^{2}\nn\\
&&~~~~~~~~~~~~~~~~~~~~~~~~+\frac{\lambda}{4}\left(u^{2}(r)-1\right)^{2}\mathrm{e}^{\eta(r)}\bigg\}\dd x.
\eer
We can get the corresponding boundary conditions from the regularity of the ansatz \eqref{4.1} and the finiteness of the energy \eqref{4.4}
\ber
&&u(0)=v(0)=0,\label{4.2}\\
&&u(\infty)=v(\infty)=1.\label{4.3}
\eer
Clearly, the Euler--Lagrange equations of \eqref{4.4} are
\ber
&&u''(r)+\frac{u'(r)}{r}-\frac{N^{2}}{r^{2}}(v(r)-1)^{2}u(r)-\frac{\lambda}{2}(u^{2}(r)-1)u(r)\mathrm{e}^{\eta(r)}=0,\label{4.5}\\
&&v''(r)-\frac{v'(r)}{r}-u^{2}(r)(v(r)-1)\mathrm{e}^{\eta(r)}=0,\label{4.6}
\eer
for any $r\in(0,\infty)$. For technical reasons, in the following, we always suppose that the vortex number $N>1$.  When $\lambda>0$, we may assume $\e^{\eta(r)}$ is smooth in $(0,\infty)$, and
\be
\e^{\eta(r)}=
\begin{cases}
 C_1 r^{-\delta}, 0<r\ll 1,\\
 C_2 r^{-\delta}, r\gg 1,
\end{cases}\nn
\ee
where $C_{1}, C_{2}$ are positive constants, $\delta=8\pi GN$ and $0<\delta\leq1$.

The existence and some fundamental properties regarding the symmetric vortex solutions governed by \eqref{4.5}--\eqref{4.6} and \eqref{4.2}--\eqref{4.3} may be stated as follows.

\begin{theorem}\label{th4.1}
For any $\lambda>0$, $0<\delta\leq1$, and integer $N>1$, the equations \eqref{4.5}--\eqref{4.6} subject to the boundary
conditions \eqref{4.2}--\eqref{4.3} have a pair of the radially symmetric vortex solution $(u,v)$ which enjoys the following properties

{\rm(i)} $0\leq u(r),v(r)\leq1$ and $u(r)\,,v(r)$ increasing, $\forall r\geq 0$ {\rm;}

{\rm(ii)} The functions $u(r),v(r)\in C^{1}\left([0,+\infty)\right)$ and $u'(0)=v'(0)=0${\rm;}

{\rm(iii)} $\lim\limits_{r\to0}u(r)=\lim\limits_{r\to0}v(r)=0$, $r^{-N}u(r)\,,r^{-2}v(r)$ decreasing and all bounded near $r=0${\rm;}

{\rm(iv)} $u(r)=1+O(r^{-\alpha}), v(r)=1+O(r^{-\beta})$ as $r\to\infty$, where $\beta>0, 1<\alpha<2+2\beta-\delta$.

\end{theorem}

We are going to attack Theorem \ref{th4.1} by the shooting method and the fixed--point theorem in a spirit similar to the work of literatures \cite{Chen1,McLeod}. First, for the given appropriate function $v$, we can show that there is a unique increasing solution $u$ to the equation \eqref{4.5} satisfying the boundary conditions $u(0)=0$, $u(\infty)=1$ by a single--parameter shooting method. Second, for the obtained solution $u$, using the shooting method again, we get a unique increasing solution $\tilde{v}$ of the equation \eqref{4.6} subjects to $\tilde{v}(0)=0$, $\tilde{v}(\infty)=1$. Next, define a mapping $T$: $T(v)=\tilde{v}$, which takes a convex bounded set into itself and is compact and continuous. Finally, applying the Schauder fixed--point theorem, we deduce that $T$ has a fixed point $v$. Consequently, we have a pair of the radially symmetric vortex solution $(u,v)$ that solves \eqref{4.5}--\eqref{4.6} under the conditions \eqref{4.2}--\eqref{4.3} and enjoys the desired properties. The details of the above steps will be carried out in subsections.

\subsection{The equation governing the $u$--component}

In this subsection, we devote to construct a solution pair of the problem \eqref{4.5}--\eqref{4.6} and \eqref{4.2}--\eqref{4.3} by considering the $u$--component equation \eqref{4.5} for a suitably given function $v$. This construction will be performed in the form of a few lemmas.

\begin{lemma}\label{lem.4.1}
Given a function $v(r)\in C\left([0,+\infty)\right)$ such that $v(r)$ increasing, $0<v(r)\leq Mr^{2}$ for all $r\leq1$ and $v(0)=0,\,v(\infty)=1$, where $M$ is a positive constant. Then we can find a unique continuously differentiable solution $u(r)$ satisfying $\eqref{4.5}$ and the conditions as follow

{\rm(i)}~$u(0)=0,\,u(\infty)=1$, $u'(0)=0$, $r^{-N}u(r)$ decreasing in $r$ and bounded near $r=0${\rm;}

{\rm(ii)}~$r^{-N}u(r)\leq M^{*}$, $\forall r\leq1$, where $M^{*}=M^{*}(M,K,N,\tilde{C},C_{1})>0$, the integer $N>1$ and $K, \tilde{C}, C_{1}$ are positive constants{\rm;}

{\rm(iii)}~$u(r)$ satisfies the asymptotic estimates
\bea
u(r)&=&ar^N+O(r^{N+2-\delta}),~r\to 0,~0<a<(K\tilde{M})^{\frac{N}{2}},~\tilde{M}=\max\{2MN^{2},\tilde{C}\},\label{4.7}\\
u(r)&=&1+O(r^{-\alpha}),~~~r\to\infty,~~\beta>0,~~1<\alpha<2+2\beta-\delta.\label{4.8}
\eea
\end{lemma}

To proceed further, we consider the existence and uniqueness of the local solution to the initial value problem that can be written in the form
\begin{eqnarray}
\label{4.9}
u''(r)+\frac{u'(r)}{r}-\frac{N^{2}}{r^{2}}(v(r)-1)^{2}u(r)-\frac{\lambda}{2}(u^{2}(r)-1)u(r)\mathrm{e}^{\eta(r)}=0,~~u(0)=0,~~\forall r>0.
\end{eqnarray}
Note that the two linearly independent solutions of equation
\begin{eqnarray}\label{4.10}
u''+\frac{u'}{r}-\frac{N^{2}}{r^{2}}u=0,\,\,\,\,r>0
\end{eqnarray}
are $r^{N}$ and $r^{-N}$. In view of $u(0)=0$, we claim that the solution of equation $\eqref{4.9}$ must satisfy $u(r)=O(r^{N})(r\rightarrow0)$. For $v(r)$ given in Lemma \ref{lem.4.1}, we see that $u(r)$ may locally be labeled by an initial parameter $a$ as follows
\begin{eqnarray}
\label{4.11}
u''(r,a)+\frac{u'(r,a)}{r}-\frac{N^{2}}{r^{2}}(v-1)^{2}u(r,a)-\frac{\lambda}{2}(u^{2}(r,a)-1)u(r,a)\mathrm{e}^{\eta(r)}=0,
\end{eqnarray}
\begin{eqnarray}
\label{4.12}
u(r,a)=ar^{N}+O\left(r^{N+2-\delta}\right),~r\rightarrow0,~0<\delta\leq1,
\end{eqnarray}
where $u(r,a)$ could be used to denote the dependence of the local solution to $\eqref{4.9}$ on a certain constant $a>0$. Such parameter $a$ may be viewed as a shooting parameter in our method. Next, we will prove the judgment that $\eqref{4.9}$ has a unique local solution near $r=0$ with $\eqref{4.12}$ for any fixed $a>0$ and $v(r)$ given.

\begin{lemma}\label{lem.4.2}
For any fixed $a>0$ and $v(r)$ given in Lemma \ref{lem.4.1}, equation $\eqref{4.9}$ has a unique local solution near $r=0$ with
\begin{eqnarray}
\label{4.13}
u(r,a)=ar^{N}+O\left(r^{N+2-\delta}\right),~~0<\delta\leq1.
\end{eqnarray}
\end{lemma}

\begin{proof} Using the basic theory of ordinary differential equations, two linearly independent solutions $r^{N}$ and $r^{-N}$ of $\eqref{4.10}$ and the condition of $u(0,a)=0$, we can transform $\eqref{4.11}$ into the integral form
\begin{eqnarray}
\label{4.14}
u(r,a)=ar^{N}+\frac{1}{2N}\int_0^r\left(\frac{r^{N}}{s^{N-1}}\!-\!\frac{s^{N+1}}{r^{N}}\right)\left\{\frac{N^{2}}{s^{2}}\left(v^{2}(s)\!-\!2v(s)\right)+\frac{\lambda}{2}\mathrm{e}^{\eta(s)}\left(u^{2}(a,s)\!-\!1\right)\right\}u(a,s)\dd s.\nn\\
\end{eqnarray}
Applying the Picard iteration, the equation $\eqref{4.14}$ can be solved near $r=0$ as follows. Setting $u_{0}(r)=ar^{N}$ and
\begin{eqnarray}
\label{4.15}
g(r,s,u_{0}(s))=\frac{1}{2N}\left(\frac{r^{N}}{s^{N-1}}\!-\!\frac{s^{N+1}}{r^{N}}\right)\left\{\frac{N^{2}}{s^{2}}\left(v^{2}(s)\!-\!2v(s)\right)+\frac{\lambda}{2}\mathrm{e}^{\eta(s)}\left(u_{0}^{2}(a,s)\!-\!1\right)\right\}u_{0}(a,s),\nn
\end{eqnarray}
where $r,s>0.$
Define
\begin{equation*}
  u_{n+1}(r)=ar^{N}+\int_0^r g(r,s,u_{n}(s))ds,~~n=0,1,2,\cdots.
\end{equation*}
Since the assumption on $v(r)$, we get
\begin{eqnarray}
\label{4.16}
|u_{1}(r)-u_{0}(r)|&\leq&\int_0^{r}\left|g(r,s,u_{0}(s))\right|\dd s\notag\\[1mm]
&\leq&\frac{a}{2N}\int_0^{r}r^{N}s^{1-\delta}\left(3MN^{2}+C_{1}\lambda\right)\dd s\notag\\[1mm]
&\leq&aM_{1}r^{{N+2-\delta}},~~r\leq\min\left\{1, \sqrt[N]{\frac{1}{a}}\right\},\nn
\end{eqnarray}
where $M>0$ is a constant and $M_{1}=\frac{3MN^{2}+C_{1}\lambda}{2(2-\delta)N}$. Thus,
\begin{eqnarray}
\label{4.17}
|u_{1}(r)|\leq aM_{1}r^{N+2-\delta}+ar^{N} \leq 2ar^{N},~~r\leq\min\left\{1, \sqrt[N]{\frac{1}{a}}, \left(\frac{1}{M_{1}}\right)^{\frac{1}{2-\delta}}\right\},\nn
\end{eqnarray}
Then, for $r\leq\min\left\{1, \sqrt[N]{\frac{1}{a}}, \left(\frac{1}{M_{1}}\right)^{\frac{1}{2-\delta}}\right\}$, we obtain
\begin{eqnarray}
\label{4.18}
|u_{2}(r)-u_{0}(r)|&\leq&\int_0^{r}\left|g(r,s,u_{1}(s))\right|\dd s\notag\\[1mm]
&\leq&\frac{a}{N}\int_0^{r}r^{N}s^{1-\delta}\left(3MN^{2}+\frac{5}{2}C_{1}\lambda\right)\dd s\notag\\[1mm]
&\leq&aM_{2}r^{N+2-\delta},\nn
\end{eqnarray}
where $M_{2}=\frac{6MN^{2}+5C_{1}\lambda}{2(2-\delta)N}$ and
\begin{eqnarray}
\label{4.19}
|u_{2}(r)|\leq aM_{2}r^{N+2-\delta}+ar^{N} \leq 2ar^{N},~~r\leq\min\left\{1, \sqrt[N]{\frac{1}{a}}, \left(\frac{1}{M_{2}}\right)^{\frac{1}{2-\delta}}\right\}.\nn
\end{eqnarray}
In general, for $r\leq\min\left\{1, \sqrt[N]{\frac{1}{a}}, \left(\frac{1}{M_{2}}\right)^{\frac{1}{2-\delta}}\right\}$, we have
\begin{eqnarray}
\label{4.20}
|u_{n+1}(r)-u_{0}(r)|&\leq&\int_0^{r}\left|g(r,s,u_{n}(s))\right|\dd s\notag\\[1mm]
&\leq&\frac{a}{N}\int_0^{r}r^{N}s\left(3MN^{2}+\frac{5}{2}\lambda C_{1}s^{-\delta}\right)\dd s\notag\\[1mm]
&\leq&aM_{2}r^{N+2-\delta},~~n=1,2,\cdots.
\end{eqnarray}
and
\begin{eqnarray}
\label{4.21}
|u_{n}(r)|\leq 2ar^{N}.\nn
\end{eqnarray}
Hence, for $r\leq\min\left\{1, \sqrt[N]{\frac{1}{a}}, \left(\frac{1}{M_{2}}\right)^{\frac{1}{2-\delta}}\right\}$, there holds
\begin{eqnarray}
\label{4.22}
|u_{n+1}(r)-u_{n}(r)|&\leq&\int_0^{r}\left|g(r,s,u_{n}(s))-g(r,s,u_{n-1}(s))\right|\dd s\notag\\[1mm]
&\leq&\frac{1}{2N}\int_0^{r}\frac{r^{N}}{s^{N+\delta-1}}\left(3MN^{2}+\frac{13\lambda}{2}C_{1}\right)\left|u_{n}-u_{n-1}\right|\dd s\notag\\[1mm]
&\leq& M_{3}\int_0^{r}\frac{r^{N}}{s^{N+\delta-1}}\left|u_{n}-u_{n-1}\right|\dd s,
~~n=1,2,\cdots.
\end{eqnarray}
where $M_{3}=\frac{6MN^{2}+13C_{1}\lambda}{4N}$. Set
\begin{eqnarray}
\label{4.23}
\tau=\min\left\{1, \sqrt[N]{\frac{1}{a}}, \left(\frac{1}{2M_{3}}\right)^{\frac{1}{2-\delta}}\right\}
\end{eqnarray}
and iterating $\eqref{4.22}$, we obtain
\begin{eqnarray}
\label{4.24}
|u_{n+1}(r)-u_{n}(r)|\leq\prod_{i=1}^{n}\frac{1}{(2-\delta)(i+1)}
aM_{1}r^{N+2-\delta}(M_{3}r^{2-\delta})^{n},~~r\in[0,\tau],~~n=1,2,\cdots.\nn
\end{eqnarray}
Thus $u_{0}(r)+\sum\limits_{n=0}^{\infty}(u_{n+1}(r)-u_{n}(r))$ uniformly converges in $[0,\tau]$ which implies that $\{u_{n}(r)\}$ uniformly converges in $[0,\tau]$. Therefore, the limit $u(r)=\lim\limits_{n\rightarrow\infty}u_{n}(r)$ is a solution of $\eqref{4.14}$ over $[0,\tau]$. The local uniqueness of the solution to $\eqref{4.14}$ is clear. Letting $n\rightarrow\infty$ in $\eqref{4.20}$, then \eqref{4.13} is obtained. Consequently, the existence and uniqueness of \eqref{4.9} near $r=0$ has been proved.
\end{proof}

According to the classical ordinary differential equation theory, the local solution to the equation \eqref{4.11} satisfying $u(0)=0$ can be extended to its maximum existence interval $[0,R_{a})$ and depends continuously on the parameter $a$, where either $R_{a}=\infty$ or $\lim\limits_{r\rightarrow R_{a}^{-}}u(r)=\infty$.

With the discussion of Lemma \ref{lem.4.2}, we can define
\begin{eqnarray*}
\mathcal{A}_{1}&=&\left\{a>0\mid\,\exists~\tilde{r}\in(0,R_{a})\,\,\mbox{such that}\,\,u'(\tilde{r},a)<0\,\,\mbox{and}\,\,0<u(r,a)<1, \forall r\in(0,\tilde{r}]\right\},\\[1mm]
\mathcal{A}_{2}&=&\left\{a>0\mid\,\exists~\hat{r}\in(0,R_{a})\,\,\mbox{such that}\,\,u(\hat{r},a)=1\,\,\mbox{and}\,\,u'(r,a)\geq0, \forall r\in(0,\hat{r}]\right\},\\[1mm]
\mathcal{A}_{3}&=&\left\{a>0\mid\,u'(r,a)\geq0\,\,\mbox{and}\,\,0<u(r,a)<1, \forall r>0\right\},
\end{eqnarray*}
where $u(r,a)$ is the unique local solution obtained in lemma \ref{lem.4.2} and $u'(r,a)=\frac{\partial u(r,a)}{\partial r}$.

It is obvious that
\begin{equation}
\label{4.25}
\mathcal{A}_{1}\cup\mathcal{A}_{2}\cup\mathcal{A}_{3}=(0,\infty),~~\mathcal{A}_{1}\cap\mathcal{A}_{2}=\mathcal{A}_{2}\cap\mathcal{A}_{3}=\mathcal{A}_{3}\cap\mathcal{A}_{1}=\emptyset.
\end{equation}

In the following, we will prove that the existence of the global solution to equation \eqref{4.5} satisfying the boundary conditions $u(0)=0$ and $u(\infty)=1$. To this end, we need to show that $\mathcal{A}_{3}$ is nonempty. In view of \eqref{4.25}, it is sufficient to show that both $\mathcal{A}_{1}$ and $\mathcal{A}_{2}$ are open and nonempty.

\begin{lemma}\label{lem.4.3}
For given $v(r)$ in Lemma \ref{lem.4.1}, the sets $\mathcal{A}_{1},\mathcal{A}_{2}$ are both open and nonempty.
\end{lemma}

\begin{proof} First, we show that $\mathcal{A}_{1}$ contains small $a$. In fact, from the proof of Lemma \ref{lem.4.2}, we get that \eqref{4.13} is true for all $a\in(0,\infty)$ and $0\leq u(r,a)<1,\,\forall r\in[0,\tau]$, where the definition of $\tau$ is in $\eqref{4.23}$. It is clear that $\tau$ is positive as $a\rightarrow0$ and we have $u(r,a)<0,\,u'(r,a)<0$ for all $r\in(0,\tau)$ when $a=0$. Therefore, for given $r_1\in(0,\tau)$, using the continuous dependence of solution on the parameter $a$, there is a small $\varepsilon>0$ such that $u(r_1,a)<0,\,u'(r_1,a)<0$ for all $a\in(0,\varepsilon)$. On the other hand, for given $a>0$, from $\eqref{4.14}$, we obtain $u(r,a)>0,\,u'(r,a)>0$ initially. Thus there exists a $r_2\in(0,r_1)\subset(0,\tau)$ so that $u(r_2,a)$ is a maximum. Note that $u(r_2,a)=1$ cannot happen (If otherwise, inserting $u'(r_2,a)=0$ and $u(r_2,a)=1$ into \eqref{4.11}, we have $u''(r_2,a)>0$ which contradicts $u''(r_2,a)\leq0$.), then there exists a $\tilde{r}\in(r_2,\tau]\subset(0,\tau)$ such that $u'(\tilde{r},a)<0$ and $0<u(r,a)<1$ for all $r\in(0,\tilde{r}]$. Consequently, $(0,\varepsilon)\subset\mathcal{A}_{1}$ is nonempty. According to the continuous dependence of the solution on the parameter $a$, we can see that $\mathcal{A}_{1}$ is open.

Second, we prove that $\mathcal{A}_{2}$ contains large $a$. Setting $t=a^{\frac{1}{N}}r$ and $\tilde{u}(t,a)=u(r,a)$, then $\eqref{4.11}$ can be converted into
\bea
\label{4.26}
&&\tilde{u}''(t,a)+\frac{\tilde{u}'(t,a)}{t}-\frac{N^{2}}{t^{2}}\tilde{u}(t,a)\nn\\
&=&\frac{N^{2}}{t^{2}}v(a^{-\frac{1}{N}}t)(v(a^{-\frac{1}{N}}t)-2)\tilde{u}(t,a)+\frac{\lambda\mathrm{e}^{\eta(a^{-\frac{1}{N}}t)}}{2a^{\frac{2}{N}}}(\tilde{u}^{2}(t,a)-1)\tilde{u}(t,a),
\eea
for all $t\in(0,a^{\frac{1}{N}}R_{a})$, with the initial condition $\tilde{u}(t,a)\sim t^{N}$ as $t\rightarrow0$. From Lemma \ref{4.2}, when $a$ large enough, we get
\begin{eqnarray}
\label{4.27}
0\leq u(r,a)<1,~~~~~r\in[0,\tau],\nn
\end{eqnarray}
where $\tau$ is defined as in $\eqref{4.23}$. Then, we obtain
\begin{eqnarray}
\label{4.28}
0\leq\tilde{u}(t,a)<1,~~~~~t\in[0,\tau a^{\frac{1}{N}}].\nn
\end{eqnarray}
Noting
\be
\left|\frac{\lambda\mathrm{e}^{\eta(a^{-\frac{1}{N}}t)}}{2a^{\frac{2}{N}}}(\tilde{u}^{2}(t,a)-1)\tilde{u}(t,a)\right|\sim \frac{\lambda C_{1}t^{N-\delta}}{2a^{\frac{2-\delta}{N}}}\rightarrow 0,~~t\rightarrow 0,\nn
\ee
thus
\be
\label{4.29*}
\left|\frac{\lambda\mathrm{e}^{\eta(a^{-\frac{1}{N}}t)}}{2a^{\frac{2}{N}}}(\tilde{u}^{2}(t,a)-1)\tilde{u}(t,a)\right|\leq\frac{\tilde{C}}{a^{\frac{2}{N}}},~~t\in[0,\tau a^{\frac{1}{N}}],\nn
\ee
where $\tilde{C}=\tilde{C}(\lambda)>0$ is a constant. Due to the fact that $0<v(r)\leq Mr^{2}$, $\forall r\leq 1$, we see
\begin{eqnarray}
\label{4.29}\nn
\left|\frac{N^{2}}{t^{2}}v(a^{-\frac{1}{N}}t)(v(a^{-\frac{1}{N}}t)-2)\tilde{u}(t,a)\right|\leq2\frac{N^{2}}{t^{2}}Ma^{-\frac{2}{N}}t^{2}=\frac{2MN^{2}}{a^{\frac{2}{N}}},~~t\in[0, \tau a^{\frac{1}{N}}].
\end{eqnarray}
According to
\begin{equation}
\label{4.29a}
\frac{2MN^{2}}{a^{\frac{2}{N}}}\rightarrow 0~\text{and}~\frac{\tilde{C}}{a^{\frac{2}{N}}}\rightarrow 0~\text{as}~a\rightarrow \infty,
\end{equation}
we know that $\eqref{4.26}$ is equal to
\begin{eqnarray}
\label{4.30}
\tilde{u}''(t)+\frac{\tilde{u}'(t)}{t}-\frac{N^{2}}{t^{2}}\tilde{u}(t)=0,
\end{eqnarray}
with $\tilde{u}(t)\sim t^{N}$ as $t\rightarrow0$. It is clear that $\tilde{u}(t)=t^{N}$ is a solution of \eqref{4.30}. Hence $u(t_{0})>1$ when $t_0\to 1^+$ and $u'(t)>0$ for all $t>0$. Thus $\mathcal{A}_{2}$ is nonempty. Using the continuous dependence of the solution on the parameter $a$, we arrive at that $\mathcal{A}_{2}$ is open.
\end{proof}

Since $\mathcal{A}_{1}$ and $\mathcal{A}_{2}$ are two open disjoint nonempty sets, by connectedness, there exists at least one $a\in\mathcal{A}_{3}$ (denoted by $a_{0}$) such that $u(r,a)$ satisfies $u'(r,a_0)\geq0$ and $0<u(r,a_0)<1$ for all $r>0$. Next, we will show that $u(r,a_0)$ is a solution of \eqref{4.5} satisfying $u(\infty,a_0)=1$.

\begin{lemma}\label{lem.4.4}
For $v(r)$ given as in Lemma \ref{lem.4.1} and $a_{0}\in\mathcal{A}_{3}$, $u(r,a_{0})$ is a solution to \eqref{4.5} satisfying $u(\infty,a_0)=1$.
\end{lemma}

\begin{proof} For $a_{0}\in\mathcal{A}_{3}$, it is easy to see that $\lim\limits_{r\rightarrow\infty}u(r,a_{0})\equiv U_{0}$. We assert $U_{0}=1$. Otherwise, suppose $U_{0}<1$. Set $F(r)=r^{\frac{1}{2}}u(r,a_{0})$, then $F(r)$ satisfies
\be
F''=\left(\frac{N^{2}}{r^{2}}(v-1)^{2}-\frac{1}{4r^{2}}+\frac{\lambda}{2}(u^{2}-1)\mathrm{e}^{\eta(r)}\right)F,\nn
\ee
and there exists a sufficiently large $R>0$ such that $F''\leq-C_{0}F$ for $r>R$, where $C_{0}>0$ is a constant. This indicates that $F(r)$ oscillates when $r>R$. Hence, there is a certain $\bar{r}>R$ so that $u'(\bar{r},a_{0})<0$ which contradicts the fact that $a_{0}\in\mathcal{A}_{2}$. Therefore, $U_{0}=1$ and the assertion of the lemma follows.
\end{proof}

We now give a series of properties of the obtained solution $u(r,a_0)$ in Lemma \ref{lem.4.2}--\ref{lem.4.4}.

\begin{lemma}\label{lem.4.5}
If $a_{0}\in\mathcal{A}_{3}$, then the solution $u(r,a_{0})$ of \eqref{4.5} satisfying $u(0,a_0)=0$ and $u(\infty,a_0)=1$ is unique.
\end{lemma}

\begin{proof} Suppose otherwise that there are two different solutions $u_1(r)$ and $u_2(r)$. Let $U(r)=u_2(r)-u_1(r)$, then $U(r)$ satisfies the boundary conditions $U(0)=U(\infty)=0$ and
\begin{eqnarray}
\label{4.32}
U''+\frac{U}{r}-\frac{N^{2}}{r^{2}}(v-1)^{2}U-\frac{\lambda}{2}\mathrm{e}^{\eta(r)}(u_{1}^{2}+u_{2}^{2}+u_{1}u_{2}-1)U=0,~~\forall r\in(0,\infty).
\end{eqnarray}
From $\eqref{4.5}$, we have
\begin{eqnarray}
\label{4.33}
u_{1}''+\frac{u_{1}'}{r}-\frac{N^{2}}{r^{2}}(v-1)^{2}u_{1}-\frac{\lambda}{2}\mathrm{e}^{\eta(r)}(u_{1}^{2}-1)u_{1}=0,~~\forall r\in(0,\infty).
\end{eqnarray}
Since $u_{2}^{2}+u_{1}u_{2}>0$, applying the Sturm comparison theorem \cite{Wa} (see chapter VI) to $\eqref{4.32}$ and $\eqref{4.33}$, we see that $u_{1}(r)$ have more zero points than $U(r)$. Note that $U(0)=U(\infty)=0$, we have $u_{1}(r_{0})=0$ for some $r_{0}\in(0,+\infty)$. This contradicts that $u_{1}(r)>0$ for all $r>0$, and concludes the proof.
\end{proof}

In view of \eqref{4.29a}, we can find a sufficiently large positive constant $K$ independent of $\tilde{u}(r)$ such that $\tilde{u}(r,a)>1$ at $r_{0}=a^{-\frac{1}{N}}t_0$ when
\be
\frac{2MN^{2}}{a^{\frac{2}{N}}}<\frac{1}{K}~~\text{and}~~\frac{\tilde{C}}{a^{\frac{2}{N}}}<\frac{1}{K}.\nn
\ee
In other words, if
\be
a^{\frac{2}{N}}>K\tilde{M},~~,\tilde{M}=\max\{2MN^{2},\tilde{C}\},\nn
\ee
then $a\in\mathcal{A}_{2}$. Since $a_{0}\notin\mathcal{A}_{2}$, thus $a_{0}^{\frac{2}{N}}\leq K\tilde{M}$.

\begin{lemma}\label{lem.4.6}
If $a_{0}\in\mathcal{A}_{3}$, then there holds
\begin{eqnarray*}
|r^{-N}u(r,a_{0})|\leq M^{*},~~\forall r\leq1,
\end{eqnarray*}
where $M^{*}=M^{*}(C_{1}, K, \tilde{M})>0$ is a constant. Besides, we have $u'(0)=0$.
\end{lemma}

\begin{proof} Using $\eqref{4.14}$, $a_{0}^{\frac{2}{N}}\leq K\tilde{M}$ and the assumption on $v(r)$, we get
\begin{eqnarray}
\label{4.36}
\left|\frac{u(r,a_0)}{r^{N}}\right|
&\leq&a_{0}+\frac{1}{2N}\int_0^r\frac{1}{s^{N-1}}\left|\frac{N^{2}}{s^{2}}2Ms^{2}+\lambda C_{1}s^{-\delta}\right|a_{0}s^{N}\dd s\notag\\[1mm]
&\leq&a_{0}+\frac{a_{0}}{2N}\int_0^r s^{1-\delta}\left|2MN^{2}+C_{1}\lambda\right|\dd s\nn\\[2mm]
&\leq&(K\tilde{M})^{\frac{N}{2}}\left(1+\frac{2MN^{2}+C_{1}\lambda}{2N}\right)\equiv M^{*},~~\forall r\leq1.\nn
\end{eqnarray}
It is obvious that $r^{-N}u(r)$ decreasing near $r=0$. Applying the above results, we get $u(r)=O(r^{N})(r\rightarrow0)$. Furthermore, $u'(0)=0$.
\end{proof}

\begin{lemma}\label{lem.4.7}
If $a_{0}\in\mathcal{A}_{3}$, then
\be\label{4.37}
u(r,a_{0})=1+O(r^{-\alpha}),~~r\rightarrow\infty,\nn
\ee
where $1<\alpha<2+2\beta-\delta$, $\beta>0$ and $0<\delta\leq1$.
\end{lemma}

\begin{proof} Set $\omega(r)=u(r)-1$, then $\eqref{4.5}$ gives us the relation
\be
\label{4.38}\nn
\omega''+\frac{\omega'}{r}-\frac{N^{2}}{r^{2}}(v-1)^{2}(\omega+1)-\frac{\lambda}{2}\mathrm{e}^{\eta(r)}\omega(\omega+1)(\omega+2)=0.
\ee
Taking the comparison function
\be
\sigma(r)=D_{1}r^{-\alpha},~~D_{1}>0,~~1<\alpha<2+2\beta-\delta,\nn
\ee
we have
\begin{eqnarray}
\label{4.40}
(\omega+\sigma)''+\frac{(\omega+\sigma)'}{r}\leq\left[\frac{N^{2}}{r^{2}}(v-1)^{2}+\frac{\lambda}{2}\mathrm{e}^{\eta(r)}(\omega+1)(\omega+2)\right](\omega+\sigma),~r\geq r_{1},
\end{eqnarray}
where $r_{1}>0$ is sufficiently large. For such fixed $r_{1}$, we can choose $D_{1}>0$ large enough to make $(\omega+\sigma)(r_{1})\geq0$. In view of this and the boundary condition $(\omega+\sigma)(r)\rightarrow0~(r\rightarrow\infty)$, we get by using the maximum principle in \eqref{4.40} that $-D_{1}r^{-\alpha}\leq\omega<0~(r>r_{1})$ as expected in \eqref{4.8}.
\end{proof}

\subsection{The equation governing the $\tilde{v}$--component}

In this subsection, we pay attention to establish the existence and uniqueness of the solution of the $\tilde{v}$--component equation.

\begin{lemma}\label{lem.4.8}
Given $v(r)$ and the associated function $u(r)$ in Lemma \ref{lem.4.1}, we can find a unique continuously differentiable solution $\tilde{v}(r)$ satisfying
\be\label{4.41}
\tilde{v}''(r)-\frac{\tilde{v}'(r)}{r}-u^{2}(r)(\tilde{v}(r)-1)\mathrm{e}^{\eta(r)}=0,~~~r>0,
\ee
subjects to the boundary conditions
\be\label{4.42}
\tilde{v}(0)=0,~~~\tilde{v}(\infty)=1,\nn
\ee
along with $\tilde{v}(r)$ increasing and $r^{-2}\tilde{v}(r)$ decreasing in $r$ and bounded near $r=0$. Besides, $r^{-2}\tilde{v}(r)\leq M^{**}$ for all $r\leq1$, where $M^{**}=M^{**}(L,\hat{C},M^{*},a_{0},N)>0$, the integer $N>1$ and $M^{*}, a_{0}, L, \hat{C}$ are the positive constants. Furthermore, we have the sharp asymptotic estimates
\bea
\tilde{v}(r)&=&br^{2}+O(r^{2+2N-\delta}),~~r\to 0,~~0<b<L\hat{C}{M^{*}}^{2},\label{4.43}\nn\\
\tilde{v}(r)&=&1+O(r^{-\beta}),~~r\to\infty,~~\beta>0.\label{4.44}\nn
\eea
\end{lemma}

First, we show that the existence and uniqueness of the local solution to the initial value problem
\begin{eqnarray}
\label{4.45}
\tilde{v}''(r)-\frac{\tilde{v}'(r)}{r}-u^{2}(r)(\tilde{v}(r)-1)\mathrm{e}^{\eta(r)}=0,~~\tilde{v}(0)=0,~~\forall r>0.
\end{eqnarray}
Noting that the leading--order part of \eqref{4.45} could be written as
\begin{eqnarray}
\label{4.46}
\tilde{v}''-\frac{\tilde{v}'}{r}=0,\,\,\,\,r>0\nn
\end{eqnarray}
for which the two linearly independent solutions are $r^{2}$ and $r^{0}$. Since $\tilde{v}(0)=0$, we can show that the solution of equation $\eqref{4.45}$ must satisfy $\tilde{v}(r)=O(r^{2})(r\rightarrow0)$. Furthermore, differential equation $\eqref{4.45}$ can be transformed into the integral form
\begin{eqnarray}
\label{4.47}
\tilde{v}(r,b)=br^{2}+\frac{1}{2}\int_0^r(r^{2}s^{-1}-s)\left[u^{2}(s)(\tilde{v}(s,b)-1)\mathrm{e}^{\eta(s)}\right]\dd s,
\end{eqnarray}
which can be solved by Picard iteration for a certain constant $b>0$, at least near $r=0$. Similar to the approach of Lemma \ref{lem.4.2}, we see that the initial value problem $\eqref{4.45}$ has a locally unique continuous solution
\be
\label{4.47a}
\tilde{v}(r,b)=br^{2}+O(r^{2+2N-\delta}),~~r\rightarrow0,
\ee
where $0<\delta\leq1$ and the integer $N>1$. Applying the continuous dependence of the solution on the parameters theorem, we know that the solution $\tilde{v}(r,b)$ depends continuously on the parameter $b$ which may be viewed as a shooting parameter in our method. A standard theorem on continuation of solutions assures us that the local solution can be extended to its maximum existence interval $[0,R_{b})$, where either $R_{b}=\infty$ or $\lim\limits_{r\rightarrow R_{b}^{-}}\tilde{v}(r)=\infty$.

With the analysis above, we are interested in $b>0$ and define
\begin{eqnarray*}
\mathcal{B}_{1}&=&\left\{b>0\mid\,\exists~\tilde{r}\in(0,R_{b})\,\,\mbox{such that}\,\,\tilde{v}'(\tilde{r},b)<0\,\,\mbox{and}\,\,0<\tilde{v}(r,b)<1, \forall r\in(0,\tilde{r}]\right\},\\[1mm]
\mathcal{B}_{2}&=&\left\{b>0\mid\,\exists~\hat{r}\in(0,R_{b})\,\,\mbox{such that}\,\,\tilde{v}(\hat{r},b)=1\,\,\mbox{and}\,\,\tilde{v}'(r,b)\geq0, \forall r\in(0,\hat{r}]\right\},\\[1mm]
\mathcal{B}_{3}&=&\left\{b>0\mid\,\tilde{v}'(r,b)\geq0\,\,\mbox{and}\,\,0<\tilde{v}(r,b)<1, \forall r>0\right\}.
\end{eqnarray*}
Clearly,
\begin{equation}
\label{4.48}\nn
\mathcal{B}_{1}\cup\mathcal{B}_{2}\cup\mathcal{B}_{3}=(0,\infty),~~\mathcal{B}_{1}\cap\mathcal{B}_{2}=\mathcal{B}_{2}\cap\mathcal{B}_{3}=\mathcal{B}_{3}\cap\mathcal{B}_{1}=\emptyset.
\end{equation}

Second, we prove the existence of the global solution to \eqref{4.41} satisfying $\tilde{v}(0)=0$ and $\tilde{v}(\infty)=1$.

\begin{lemma}\label{lem.4.9}
For given $v(r)$ and the associated function $u(r)$ in Lemma \ref{lem.4.1}, the sets $\mathcal{B}_{1},\mathcal{B}_{2}$ are both open and nonempty.
\end{lemma}

\begin{proof} Firstly, we show that $\mathcal{B}_{1}$ contains small $b$. Indeed, there is a $\kappa$ such that $0\leq \tilde{v}(r,b)<1$ for all $r\in[0,\kappa]$. It is obvious that $\kappa$ is positive as $b\rightarrow0$ and we obtain $\tilde{v}(r,b)<0,\,\tilde{v}'(r,b)<0$ for all $r\in(0,\kappa)$ as $b=0$. Hence, for given $r_1\in(0,\kappa)$, applying the continuous dependence of solution on the parameter $b$, there exists a small $\varepsilon>0$ so that $\tilde{v}(r_1,b)<0,\,\tilde{v}'(r_1,b)<0$ for all $b\in(0,\varepsilon)$. Besides, for given $b>0$, from $\eqref{4.47}$, we have $\tilde{v}(r,b)>0,\,\tilde{v}'(r,b)>0$ initially. Therefore, there exists a $r_2\in(0,r_1)\subset(0,\kappa)$ such that $\tilde{v}(r_2,b)$ is a maximum. Noting that $\tilde{v}(r_2,b)=1$ cannot happen. If otherwise, in view of \eqref{4.41}, we get $\tilde{v}''(r_{2},b)=0$ which implies $\tilde{v}(r_{2},b)\equiv1$, contradicting \eqref{4.47a}. Then there is a $\tilde{r}\in(r_2,\kappa]\subset(0,\kappa)$ such that $\tilde{v}'(\tilde{r},b)<0$ and $0<\tilde{v}(r,b)<1$ for all $r\in(0,\tilde{r}]$. Thus, $(0,\varepsilon)\subset\mathcal{B}_{1}$ is nonempty. By the continuous dependence of the solution on the parameter $b$, we see that $\mathcal{B}_{1}$ is open.

Secondly, we prove that $\mathcal{B}_{2}$ contains large $b$. To show this, make the change of scale $t=b^{\frac{1}{2}}r$ and $\hat{v}(t,b)=\tilde{v}(r,b)$ in \eqref{4.41} so that
\begin{eqnarray}
\label{4.50}
\hat{v}''(t,b)-\frac{\hat{v}'(t,b)}{t}=\frac{u^{2}(b^{-\frac{1}{2}}t)(\hat{v}(t,b)-1)\mathrm{e}^{\eta(b^{-\frac{1}{2}}t)}}{b},
\end{eqnarray}
for all $t\in(0,b^{\frac{1}{2}}R_{b})$, with the initial condition $\hat{v}(t,b)\sim t^{2}$ as $t\rightarrow0$. In view of
\begin{eqnarray}
0\leq\hat{v}(t,b)<1,~~~~~t\in[0, \kappa b^{\frac{1}{2}}]\nn
\end{eqnarray}
and the fact that $0<u(r)\leq M^{*}r^{N}$ for all $r\leq 1$, we have
\begin{eqnarray}
\label{4.51}
\left|\frac{u^{2}(b^{-\frac{1}{2}}t)(\hat{v}(t,b)-1)\mathrm{e}^{\eta(b^{-\frac{1}{2}}t)}}{b}\right|\leq\frac{\hat{C}{M^{*}}^{2}}{b},~~\forall t\in[0, \kappa b^{\frac{1}{2}}],
\end{eqnarray}
where $\hat{C}, M^{*}>0$ are constants. Let $b\rightarrow\infty$, using \eqref{4.51}, then the equation \eqref{4.50} is equal to
\begin{eqnarray}
\label{4.52}
\hat{v}''(t,b)-\frac{\hat{v}'(t,b)}{t}=0,
\end{eqnarray}
with $\hat{v}(t)\sim t^{2}$ as $t\rightarrow0$. It is clear that $\hat{v}(t)=t^{2}$ is a solution of \eqref{4.52}. Hence $\hat{v}(t_{0})>1$ when $t_{0}\to 1^+$ and $\hat{v}'(t)>0$ for all $t>0$. Thus, $\mathcal{B}_{2}$ is nonempty. Moreover, there is some sufficiently large positive constant $L$ independent of the choice of $\hat{v}(r)$ such that when
\be\label{4.53}
\frac{\hat{C}{M^{*}}^{2}}{b}<\frac{1}{L},\nn
\ee
we have $b\in\mathcal{B}_{2}$. The openness of $\mathcal{B}_{2}$ can be inferred from the continuous dependence of the solution on the parameter $b$.
\end{proof}

By connectedness, there exists a $b\in\mathcal{B}_{3}$, say, $b_{0}$, such that $\tilde{v}(r,b)$ satisfies conditions $\tilde{v}'(r,b_0)\geq0$ and $0<\tilde{v}(r,b_0)<1$ for all $r>0$, and we have
\begin{equation}
\label{4.55}
b_{0}\leq L\hat{C}{M^{*}}^{2}.
\end{equation}

\begin{lemma}\label{lem.4.10}
For given $v(r)$ and the associated function $u(r)$ in Lemma \ref{lem.4.1}. If $b_{0}\in\mathcal{B}_{3}$, then $\tilde{v}(r,b_{0})$ is a solution to \eqref{4.41} satisfying $\tilde{v}(\infty,b_0)=1$.
\end{lemma}

\begin{proof} For $b_{0}\in\mathcal{B}_{3}$, we can get $\lim\limits_{r\rightarrow\infty}\tilde{v}(r,b_{0})\equiv V_{0}$. We need only to prove $V_{0}=1$. If otherwise $V_{0}<1$, then \eqref{4.41} yields $(r\tilde{v}'(r,b_{0}))'=C_{2}r^{1-\delta}(V_{0}-1+o(1))$ which implies that
\be
\tilde{v}'(r,b_{0})=C_{2}\left[\frac{1}{2-\delta}r^{1-\delta}(V_{0}-1)+o(r^{1-\delta})\right],~~r\rightarrow\infty,\nn
\ee
where $0<\delta\leq1$ and $C_{2}$ is a positive constant. This is impossible since $\tilde{v}'(r,b_{0})>0$ for all $r>0$. Therefore $V_{0}=1$.
\end{proof}

In the following, some properties related to the solution $\tilde{v}(r,b_{0})$ obtained in Lemma \ref{lem.4.8}--\ref{lem.4.10} are given.

\begin{lemma}\label{lem.4.11}
If $b_{0}\in\mathcal{B}_{3}$, then the solution $\tilde{v}(r,b_{0})$ of \eqref{4.41} satisfying $\tilde{v}(0,b_0)=0$ and $\tilde{v}(\infty,b_0)=1$ is unique.
\end{lemma}

\begin{proof} Assume that there are two different solutions $\tilde{v}_{1}(r)$ and $\tilde{v}_{2}(r)$. Set $V(r)=\tilde{v}_{1}(r)-\tilde{v}_{2}(r)$, then $V(r)$ satisfies the boundary conditions $V(0)=V(\infty)=0$ and the equation
\begin{eqnarray}
\label{4.54}\nn
V''(r)-\frac{V'(r)}{r}=u^{2}(r)\mathrm{e}^{\eta(r)}V(r),~~\forall r\in(0,\infty).
\end{eqnarray}
So the maximum principle implies that $V(r)\equiv0$, $\forall r\geq0$.
\end{proof}

\begin{lemma}\label{lem.4.12}
If $b_{0}\in\mathcal{B}_{3}$, then
\begin{eqnarray*}
|r^{-2}\tilde{v}(r,b_{0})|\leq M^{**},~~\forall r\leq1,
\end{eqnarray*}
where $M^{**}=M^{**}(L,\hat{C},M^{*},a_{0},N)$ is a positive constant. Besides, we have $\tilde{v}'(0)=0$.
\end{lemma}

\begin{proof} In view of \eqref{4.47}, \eqref{4.55} and Lemma \ref{lem.4.1}, we arrive at
\begin{eqnarray}
\label{4.56}
\left|\frac{\tilde{v}(r,b_0)}{r^{2}}\right|
&\leq&b_{0}+\frac{1}{2}\int_0^r s^{-1}\left|u^{2}(s)(\tilde{v}(s,b_0)-1)\mathrm{e}^{\eta(s)}\right|\dd s\notag\\[1mm]
&\leq&b_{0}+a_{0}^{2}\int_0^r s^{2N-1-\delta}\dd s\nn\\[2mm]
&\leq&L\hat{C}{M^{*}}^{2}+\frac{a_{0}^{2}}{2N-1}\equiv M^{**},~~\forall r\leq1.\nn
\end{eqnarray}
It is easy to see that $r^{-2}\tilde{v}(r)$ decreasing near $r=0$. Consequently, we have $\tilde{v}'(0)=0$.
\end{proof}

\begin{lemma}\label{lem.4.13}
If $b_{0}\in\mathcal{B}_{3}$, then we have the sharp estimate
\be\label{4.57}
\tilde{v}(r,b_{0})=1+O(r^{-\beta}),~~r\rightarrow\infty,\nn
\ee
where $\beta$ is a positive constant.
\end{lemma}

\begin{proof} Set $\nu(r)=\tilde{v}(r)-1$, we may convert $\eqref{4.41}$ into
\be
\label{4.58}
\nu''(r)-\frac{\nu'(r)}{r}-u^{2}(r)\nu\mathrm{e}^{\eta(r)}=0.\nn
\ee
Taking the comparison function
\be
\sigma(r)=D_{2}r^{-\beta},~~D_{2}>0,~~\beta>0,\nn
\ee
then there exists a sufficiently large $r_{1}>0$ such that
\begin{eqnarray}
\label{4.60}
(\nu+\sigma)''-\frac{(\nu+\sigma)'}{r}\leq u^{2}\mathrm{e}^{\eta(r)}(\nu+\sigma),~~r>r_{1}.\nn
\end{eqnarray}
We can choose $D_{2}>0$ large enough to make $(\nu+\sigma)(r_{1})\geq0$. According to the boundary condition $(\nu+\sigma)(r)\rightarrow0~(r\rightarrow\infty)$, we see $-D_{2}r^{-\beta}\leq\nu<0~(r>r_{1})$ by applying the maximum principle theorem. Therefore, the decay estimate for $\tilde{v}(r,b_{0})$ at infinity is established.
\end{proof}

\subsection{Proof of Theorem \ref{th4.1}}

We are now prepared to prove Theorem \ref{th4.1}. That is, we are going to use the Schauder fixed point theorem to establish the existence of the symmetric vortex solutions arising from the magnetic Ginzburg--Landau theory in curved space.

Define a Banach space $\mathscr{X}$
\be\nn
\mathscr{X}=\left\{v(r)\,|\,v(r)\in C([0,+\infty)),r^{-k}\left(1+r^{k}\right)v(r)\,\mbox{is bounded},\,0<k<2\right\}
\ee
with the norm
\begin{eqnarray*}
\left\|v(r)\right\|_{\mathscr{X}}=\sup_{r\in[0,+\infty)}\left|r^{-k}\left(1+r^{k}\right)v(r)\right|,
\end{eqnarray*}
and the non--empty, bounded, closed, convex subset $\mathcal{S}$ of $\mathscr{X}$
\begin{eqnarray*}
\mathcal{S}=\big\{v(r)\in\mathscr{X}\,\big|~\,|r^{-2}v(r)|\leq \hat{M},\,\forall0<r\leq1;\,v'(r)\geq0,\,0<v(r)\leq1,\,\forall r>0;\,v(\infty)=1\},
\end{eqnarray*}
where $\hat{M}=\max\{M, M^{**}\}$.

The process of the two--step shooting argument gives us a mapping $T{\rm:}\,v\rightarrow \tilde{v}$ on $\mathcal{S}$ which takes $\mathcal{S}$ into itself by applying Lemma \ref{lem.4.1} and \ref{lem.4.8}. In the following, we will show the mapping $T$ is compact and continuous.

To prove the compactness of $T$, we will demonstrate that if $\left\{v_{n}(r)\right\}$ is an arbitrary bounded sequence in $\mathcal{S}$, then $\{\tilde{{v}_{n}}(r)\}$ have a convergent sub--sequence in $\mathcal{S}$. For any $\left\{v_{n}(r)\right\}$ in $\mathcal{S}$,
using Lemma \ref{lem.4.8} and $\{\tilde{{v}_{n}}(r)\}\subset\mathcal{S}$, we see
\begin{eqnarray*}
{\tilde{{v}_{n}}}'(r)=2b_{0}r+\int_0^r rs^{-1}\left[u_{n}^{2}(s)(\tilde{v}_{n}(s)-1)\mathrm{e}^{\eta(s)}\right]\dd s\nn
\end{eqnarray*}
is bounded on any closed subinterval of $(0,+\infty)$. That is, there exists a constant $P>0$ independent of $n$ such that $|\tilde{v_{n}}'(r)|\leq P$ on any closed subinterval of $(0,+\infty)$. By the mean value theorem, for any $\varepsilon>0,\,r_{1},r_{2}\in[\gamma,R]\subset(0,+\infty)$, there is a $l=\frac{\varepsilon}{P+1}>0$ so that $|\tilde{v_{n}}(r_{1})-\tilde{v_{n}}(r_{2})|
=|\tilde{v_{n}}'(\xi)||r_{1}-r_{2}|<\varepsilon$ when $|r_{1}-r_{2}|<l$, where $\xi$ lies between $r_{1}$ and $r_{2}$. Therefore, the equicontinuity of $\{\tilde{v_{n}}(r)\}$ is established. It is clear that $\{\tilde{v_{n}}(r)\}$ is uniformly bounded in view of $\{\tilde{v_{n}}(r)\}\subset\mathcal{S}$. For any $\left\{v_{n}(r)\right\}$ in $\mathcal{S}$, applying the Arzela--Ascoli theorem, there exists a subsequence of $\{\tilde{v_{n}}(r)\}$~$($ still denoted as $\{\tilde{v_{n}}(r)\})$, which uniformly converges in any compact subinterval of $(0,+\infty)$ $($denoted as $[\gamma,R])$. We may assume that $\tilde{v_{n}}(r)\rightarrow\tilde{v}(r)$ uniformly over $[\gamma,R]\subset(0,+\infty)$, thus
\begin{eqnarray}
\label{4.61}
\sup_{r\in[\gamma,R]}\left|r^{-k}(1+r^{k})(\tilde{v_{n}}(r)-\tilde{v}(r))\right|
\leq\left(1+\gamma^{-k}\right)\sup_{r\in[\gamma,R]}
\left|\tilde{v_{n}}(r)-\tilde{v}(r)\right|
\rightarrow0,~~n\rightarrow\infty.\nn
\end{eqnarray}
Then the only question is whether for the limit functions, say $\tilde{v}(r)$, it is true that
\begin{equation*}
||\tilde{v_{n}}(r)-\tilde{v}(r)||_{\mathscr{X}}\rightarrow0,~~n\rightarrow\infty.
\end{equation*}
As $r\rightarrow0$, since $\{\tilde{v}_n(r)\}\subset\mathcal{S}$ and the set $\mathcal{S}$ is closed, we have $\tilde{v}_n(r)\leq r^{2} \hat{M},$\,$\tilde{v}(r)\leq r^{2} \hat{M}$ for all $0<r\leq1$. Furthermore, for given $\varepsilon>0$, let $\gamma>0$ small enough so that $4\hat{M}\gamma^{2-k}<\varepsilon$, then
\begin{eqnarray*}
\sup_{r\in[0,\gamma)}|r^{-k}(1+r^{k})(\tilde{v}_n(r)-\tilde{v}(r))|<2\sup_{r\in[0,\gamma)}|r^{2-k}r^{-2}(\tilde{v}_n(r)-\tilde{v}(r))|<4\hat{M}\gamma^{2-k}<\varepsilon.
\end{eqnarray*}
With $\gamma>0$ fixed, we can find $n$ large enough that
\begin{eqnarray}
\label{4.62}
\sup_{r\in[0,\gamma)}\left|r^{-k}(1+r^{k})(\tilde{v}_n(r)-\tilde{v}(r))\right|<\varepsilon.\nn
\end{eqnarray}
As $r\rightarrow+\infty$, in view of the set $\mathcal{S}$ is closed and the asymptotic estimate in Lemma \ref{lem.4.13}, we get $|\tilde{v}_n(r)-\tilde{v}(r)|\rightarrow0$ as $r\rightarrow\infty$. Thus, for $R>0$ sufficiently large, we obtain $\sup_{r\in(R,+\infty)}|r^{-k}(1+r^{k})(\tilde{v}_n(r)-\tilde{v}(r))|\rightarrow0$ uniformly as $n\rightarrow\infty$. In conclusion, $||\tilde{v_{n}}(r)-\tilde{v}(r)||_{\mathscr{X}}\rightarrow0$ as $n\rightarrow\infty$ and the mapping $T$ is compact.

To show $T$ is continuous, we shall prove that if $\left\|v_{1}-v_{2}\right\|_{\mathscr{X}}\rightarrow0$, then  $\left\|Tv_{1}-Tv_{2}\right\|_{\mathscr{X}}\rightarrow0$ for any $v_{1},v_{2}\in\mathcal{S}$. Since $Tv_{1}(\infty)=Tv_{2}(\infty)=1$, then for any $\varepsilon>0$, there is an adequately large $R>0$ (independent of $v_{1}$ and $v_{2}$) so that
\begin{eqnarray}
\label{4.63}
\sup_{r\in[R,+\infty)}\left|r^{-k}\left(1+r^{k}\right)\left(Tv_{1}-Tv_{2}\right)\right|<\varepsilon.\nn
\end{eqnarray}
For the above given $\varepsilon$, in view of $Tv_{1}(0)=Tv_{2}(0)=0$, there exists $\gamma>0$ small such that
\begin{eqnarray}
\label{4.64}
\sup_{r\in(0,\gamma]}\left|r^{-k}\left(1+r^{k}\right)\left(Tv_{1}-Tv_{2}\right)\right|<\varepsilon.\nn
\end{eqnarray}
According to the continuity of $T$ over $[\gamma,R]$, we arrive at $T$ is continuous on $\mathscr{X}$.

Finally, applying the Schauder fixed--point theorem, there exists a pair of solution $(u,v)$ to the two--point boundary value problem \eqref{4.5}--\eqref{4.6} and \eqref{4.2}--\eqref{4.3}. Hence the conclusion follows exactly as in Theorem \ref{th4.1}.


\begin{thebibliography}{99}

\bibitem{Ag}
D. F. Agterberg, E. Babaev, and J. Garaud, Microscopic prediction of skyrmion lattice state in clean interface superconductors, {\em Phys. Rev.} B {\bf 90} (2014) 064509.


\bibitem{B}
J. Bardeen, L. N. Cooper, and J. R. Schrieffer, Microscopic theory of superconductivity, {\em Phys. Rev.} {\bf 106} (1957) 162-164.

\bibitem{Ba}
J. Bardeen, L. N. Cooper, and J. R. Schrieffer, Theory of superconductivity, {\em Phys. Rev.} {\bf 108} (1957) 1175-1204.

\bibitem{Be}
F. Bethuel, H. Brezis, and F. H\'{e}lein, {\em Ginzburg--Landau vortices}, Birkh\"{a}user, Basel, 1994.

\bibitem{Br}
R. H. Brandenberger, Cosmic strings and the large--scale structure of the universe, {\em Phys. Scr.} T {\bf36} (1991) 114-126.

\bibitem{Cao}
L. Cao and S. Chen, An existence theorem for generalized Abelian Higgs equations and its application, {\em Proc. Roy. Soc. A} {\bf 480} (2024) 2298.

\bibitem{Ch}
S. A. Chen and K. T. Law, Ginzburg--Landau theory of flat--band superconductors with quantum metric, {\em Phys. Rev. Lett.} {\bf 132} (2024) 026002.

\bibitem{Chen1}
S. Chen and Y. Xu, The existence of dyon solutions for generalized Weinberg--Salam model, \emph{J. Math. Phys.} {\bf 64} (2023) 041503.

\bibitem{De}
R. Der and M. Herrmann, Critical phenomena in self--organizing feature maps: Ginzburg-Landau approach, {\em Phys. Rev.} E {\bf 49} (1994) 5840.


\bibitem{Fa}
F. Falk, Ginzburg--Landau theory of static domain walls in shape--memory alloys, {\em  Z. Phys.} B {\bf 51} (1983) 177-185.

\bibitem{F}
D. Flores--Alfonso, C. S. Lopez--Monsalvo, and M. Maceda, Contact geometry in superconductors and New Massive Gravity, {\em Phys. Lett.} B {\bf 815} (2021) 136143.

\bibitem{G}
L. P. Gor'kov, Microscopic derivation of the Ginzburg-Landau equations in the theory of superconductivity, {\em Sov. Phys. JETP} {\bf 9} (1959) 1364-1367.

\bibitem{Gu}
S. S. Gubser, Breaking an Abelian gauge symmetry near a black hole horizon, {\em Phys. Rev.} D {\bf 78} (2008) 065034.

\bibitem{H}
K.-H. Hoffmann and Q. Tang, {\em Ginzburg--Landau phase transition theory and superconductivity}, Birkh\"{a}user, Basel, 2012.


\bibitem{Kib}
T. W. B. Kibble, Some implications of a cosmological phase transition, {\em Phys. Rep.} {\bf67} (1980) 183-199.



\bibitem{Li}
S. Lin and S. Hayami, Ginzburg--Landau theory for skyrmions in inversion--symmetric magnets with competing interactions, {\em Phys. Rev.} B {\bf 93} (2016) 064430.


\bibitem{Ma}
B. A. Malomed, New findings for the old problem: Exact solutions for domain walls in coupled real GinzburgLandau equations, {\em Phys. Lett.} A {\bf 422} (2021) 127802.

\bibitem{McLeod}
J. B. McLeod and C. B. Wang, Existence of solutions for the Cho--Maison monopole/dyon, \emph{Proc. Roy. Soc.} A {\bf457} (2001) 773-784.

\bibitem{M}
R. C. McOwen, Conformal metric in $\mathbb{R}^2$ with prescribed Gaussian curvature and positive total
curvature, {\em Indiana Univ. Math. J.} {\bf 34} (1985) 97-104.

\bibitem{Mi}
A. V. Milovanov and J. J. Rasmussen, Fractional generalization of the Ginzburg--Landau equation: an unconventional approach to critical phenomena in complex media, {\em Phys. Lett.} A {\bf 337} (2005) 75-80.

\bibitem{Ni}
W. M. Ni, On the elliptic equation $\triangle u+K(x)u^{(n+2)/(n-2)}=0$, its generalizations, and applications in geometry, {\em Indiana U. Math. J.} {\bf 31} (1982) 493-529.

\bibitem{Pa}
A. R. Pack, J. Carlson, S. Wadsworth, and M. K. Transtrum, Vortex nucleation in superconductors within time-dependent Ginzburg--Landau theory in two and three dimensions: Role of surface defects and material inhomogeneities, {\em Phys. Rev.} B {\bf 101} (2020) 144504.

\bibitem{Sc}
D. J. Scalapino, Y. Imry, and P. Pincus, Generalized Ginzburg--Landau theory of pseudo--one--dimensional systems, {\em Phys. Rev.} B {\bf 11} (1975) 2042.


\bibitem{Tr}
M. K. Transtrum, G. Catelani, and J. P. Sethna, Superheating field of superconductors within Ginzburg-Landau theory, {\em Phys. Rev.} B {\bf 83} (2011) 094505.

\bibitem{Ve}
Y. Verbin, Cosmic strings in the Abelian Higgs model with conformal coupling to gravity, {\em Phys. Rev.} D {\bf 59} (1999) 105015.

\bibitem{V}
A. Vilenkin, Cosmic strings and domain walls, {\em Phys. Rep.} {\bf 121} (1985) 263-315.

\bibitem{Vil}
A. Vilenkin, Cosmological density fluctuations produced by vacuum strings, {\em Phys. Rev. Lett.} {\bf46} (1981) 1496.

\bibitem{Vi}
A. Vilenkin and E. P. S. Shellard, {\em Cosmic strings and other topological defects}, Cambridge University Press, Cambridge, 2000.

\bibitem{Wa}
W. Walter, {\em Ordinary Diffential Equations}, Springer, New York, 1998.

\bibitem{Wi}
E. Witten, Superconducting strings, {\em Nucl. Phys.} B {\bf249} (1985) 557-592.

\bibitem{Ze}
Ya. B. Zeldovich, Cosmological fluctuations produced near a singularity, {\em Mon. Not. Roy. Astron. Soc.} {\bf192} (1980) 663-667.

\end{thebibliography}
\end{document}